%% file: nclag.tex
\documentclass{m2an}

\usepackage{hyperref}
\usepackage{subfigure}
\usepackage{graphicx}%
\usepackage{multirow}%
\usepackage{amsmath,amssymb,amsfonts}%
\usepackage{mathtools}%
\usepackage{amsthm}%
\usepackage{mathrsfs}%
\usepackage{xcolor}%
\usepackage{textcomp}%
\usepackage{manyfoot}%
\usepackage{booktabs}%
\usepackage{algorithm}%
\usepackage{algorithmicx}%
\usepackage{algpseudocode}%
\usepackage{listings}%
\usepackage{color,xcolor}
\definecolor{orange}{rgb}{1,0.5,0}
\definecolor{grey}{rgb}{0.65, 0.65, 0.65}
\usepackage{enumitem}

\include{abbr-notation}

\def\beqs{\begin{equation*}}
\def\eeqs{\end{equation*}}
\def\beq{\begin{equation}}
\def\eeq{\end{equation}}
\def\beql#1{\begin{equation}\label{#1}}

\begin{document}

\title{Hugoniot Relation for Multi-Temperature Euler Equations of Compressible Plasma Flows}
%
\author{Zhifang Du}\address{Institute of Applied Physics and Computational Physics, Beijing, 100094, China, e-mail: \url{du@mail.bnu.edu.cn}}
\author{Aleksey Sikstel}\address{RWTH Aachen University, Institut f\"ur Geometrie und praktische Mathematik, Aachen, 52062, Germany, e-mail: \url{sikstel@acom.rwth-aachen.de}}
\date{\today}
\begin{abstract}
Shock solutions for multi-temperature Euler equations are inherently ambiguous due to the loss of microscopic physical detail during model reduction and occurrence of non-conservative terms. This paper presents a detailed analytical study of shock structures in such models. We present two distinct Hugoniot relations, each corresponding to a physically admissible shock solution. Through classical analysis à la Courant--Friedrichs, we demonstrate that both satisfy admissibility conditions, revealing a fundamental non-uniqueness in shock structures. By relating these solutions to existing numerical schemes, the structure-preserving and vanishing viscosity approaches, we provide physically justified references for constructing and evaluating discontinuous numerical approximations. In particular, we emphasize that the Hugoniot relation is not uniquely determined by the macroscopic PDEs alone, but must be supplied from external sources such as experiments or first-principles simulations. This insight demonstrates the essential role of microscopic physics in resolving shock ambiguity and contributes to the theoretical foundation for modeling discontinuous plasma flows.
\end{abstract}
%
%
\maketitle

\section{Introduction}
In plasma flows modeled as compressible fluids, ions and electrons typically evolve at distinct temperatures. This disparity gives rise to a multi-temperature model governed by a system of compressible Euler equations with multiple internal energies. Due to this formulation, the governing equations are intrinsically non-conservative. Solving the Riemann problem for such systems is not only central to the development of numerical methods, such as finite volume and discontinuous Galerkin schemes, but also crucial for understanding the underlying mathematical and physical structure of the governing partial differential equations (PDEs).

In classical gas dynamics, shock waves are uniquely determined by the Rankine--Hugoniot jump conditions derived from the conservation of mass, momentum, and total energy. This classical framework, however, is no longer sufficient in the multi-temperature regime. The independent phasic internal-energy equations contain non-conservative work terms $p_k\operatorname{div}\bu$, whose values across a discontinuity are not defined in the classical distributional sense. Consequently, the macroscopic equations determine only a total-energy Rankine--Hugoniot condition, which leads to a Hugoniot hypersurface in the phase space. Selecting a post-shock curve on this hypersurface requires additional path or regularization information.

The Dal Maso--LeFloch--Murat (DLM) theory~\cite{dal1995} provides a weak formulation for such non-conservative systems by prescribing a family of paths in the phase space. 
Numerical schemes have been extensively studied based on this theory, see~\cite{pares2006numerical,castro2013entropy,fjordholm2012accurate,munoz2007godunov}.
Nevertheless, the central modeling issue remains physical: the chosen path or regularization must encode the closure information lost in passing from a microscopic description to the macroscopic non-conservative system.

This study makes three principal contributions to the theory of shock solutions for non-conservative PDEs modeling plasmas with inter-species thermal disequilibrium. {First, we characterize two distinct analytical Hugoniot relations, each corresponding to a physically plausible shock solution. One was first studied by Prokopov~\cite{prokopov2004riemann}, whereas the other is induced by the vanishing viscosity traveling-wave construction of Chalons and Coquel~\cite{chalons2005riemann}. }

Second, we analyze the Hugoniot relations within the classical framework of Courant and Friedrichs \cite{Courant-Friedrichs}, demonstrating that both satisfy the criteria expected of admissible Hugoniot relations. As a result, we provide explicit shock solutions that serve as reference standards for numerical schemes, especially in the context of plasma modeling.

Third, we identify the source of ambiguity between these solutions. The non-uniqueness of shock structures, as observed in \cite{strt-prsv}, stems from the model reduction process, during which critical microscopic physical information is lost. We argue that this missing information is precisely the Hugoniot relation, which cannot easily be derived from the macroscopic PDEs alone. Rather, it must be obtained from microscopic physical mechanisms: whether through physical experiments~\cite{yarger-1955, mcqueen-1957, mcqueen-1960, isbell-1965, mcqueen-optical, royal-society-2013, caep-ifp-2013} or first-principles simulations, see for instance~\cite{mattson2010first,swift2001first}. Moreover, this ambiguity directly reflects the non-uniqueness in selecting admissible paths for a consistent weak formulation within the DLM framework.

To connect the analysis with computations, we compare analytical Hugoniot relations with existing numerical schemes. The first Hugoniot relation aligns with the  structure-preserving scheme proposed in \cite{strt-prsv}, which encodes a thermodynamic symmetry among species into its numerical discretization. However, this scheme does not account for non-equilibrium interactions within the shock layer and lacks a well-justified shock structure. The second Hugoniot relation corresponds to the vanishing viscosity scheme introduced in \cite{scheme-e-vv}, where shock solutions are defined as the limit of traveling wave solutions of the Navier-Stokes equations as viscosity tends to zero. A similar approach is used in the staggered grid Lagrangian scheme \cite{sgh-vv}, although these schemes rely on viscosity coefficients that describe equilibrium transport processes, thereby failing to capture inherently non-equilibrium dynamics within shock layers.

In conclusion, this work advances the understanding of shock solutions in non-conservative systems by providing analytical Hugoniot relations, offering explicit shock structures as references for numerical methods, and emphasizing that the correct Hugoniot relation must be informed by microscopic physics or experiment data and cannot be derived from the macroscopic model alone. This realization clarifies the source of ambiguity in shock solutions and motivates the integration of experimental or first-principles data into model development and numerical scheme validation.

{The remainder of the paper is organized as follows. Section~\ref{sec:govern-pde} introduces the multi-temperature model, its entropy structure, and its characteristic fields. Section~\ref{sec:wave-right} derives and analyzes the segment and vanishing-viscosity Hugoniot relations. Section~\ref{sec:riemann-problem} constructs the corresponding Riemann solutions. Section~\ref{sec:scheme} connects the closures with numerical schemes and presents two- and three-temperature test cases. The final section summarizes the main implications and outlines directions for data-informed closures and high-order numerical schemes.}

\section{Multi-temperature Euler equations}\label{sec:govern-pde}

{The macroscopic multi-temperature description is motivated by the separation of collisional time scales in a plasma. Because electrons and heavy ions have very different masses, inter-species collisions lead to separate thermal equilibria before the equilibrium between their temperatures; see~\cite[Section 43]{lifshitz1981physical}. This separation is particularly important in shock and post-shock relaxation layers~\cite[Chapter VII]{zeldovich1966physics}. The three-temperature version of this system was intensively studied in~\cite{lowrie}.
Multi-pressure Navier--Stokes systems model this thermal nonequilibrium system through independent phasic internal energies and pressures. Closely related systems also arise in turbulence modeling; see, for example,~\cite{chalons2005navier,chalons2010time,berthon2007nonlinear,berthon2005numerical,aregba2020viscous}.}

We are interested in its inviscid limit, the multi-temperature Euler equations  
\begin{subequations}\label{eq:multi-temp-nD}
\begin{align}
\label{eq:multi-temp-nD-rho}
&\dfr{\partial \rho}{\partial t} + \divop \rho \bu = 0,\\ 
\label{eq:multi-temp-nD-rho-u}
&\dfr{\partial \rho\bu}{\partial t} + \divop\left( \rho \bu \otimes \bu + \sum_{k=1}^K p_k \bI \right) = \mathbf{0},\\
\label{eq:multi-temp-nD-rho-e_k}
&\dfr{\partial\rho e_k}{\partial t} + \divop(\rho e_k\bu) + p_k \divop \bu = 0, \quad k=1,\ldots,K,
\end{align}
\end{subequations}
for $t > 0$ and $\bx \in \mathbb{R}^d$. Here, $\rho$ denotes the density, $\bu$ the velocity, $e_{k}$ the phasic specific internal energy, $p_{k}$  the phasic pressure{, and $K$ the number of species}. Notably, system~\eqref{eq:multi-temp-nD} entails non-conservative products and cannot be reduced to conservation laws. In order to close the system we impose the perfect gas equations of state (EOS):
\begin{equation}\label{eq:eos}
p_k(\rho, e_k)=(\gamma_k-1)\rho e_k,
\end{equation}
where $\gamma_k> 1$ denotes the phasic adiabatic index. The phasic (physical)  entropy $s_k$ is obtained following the Gibbs law
\begin{equation}\label{eq:gibbs}
T_kds_k=de_k+p_kd\tau, 
\end{equation}
where $T_k$ denotes the phasic temperature and $\tau = \frac 1 \rho$.
In this particular case of a polytropic gas mixture, $s_k$ can be explicitly computed  from~\eqref{eq:gibbs} as
\begin{equation*}
s_k=\dfr{1}{\gm_k-1}\ln{\dfr{\varsigma_k}{\gm_k-1}},
\end{equation*}
where $\varsigma_k=p_k\tau^{\gm_k}$.
Furthermore, the squared  speed of sound is given by
\begin{equation*}
c^2 \coloneq\sum_{k=1}^K\left(\frac{\pt p_k}{\pt\rho}\right)_{s_k}= \sum_{k=1}^K\frac{\gm_k p_k}{\rho}.
\end{equation*}

The conservation law for the total energy $E \coloneq \frac 12 \|\bu\|_2^2+\sum_{k=1}^K e_k$ can be directly derived from the equations~\eqref{eq:multi-temp-nD} as
\begin{equation}\label{eq:consv-energy}
\dfr{\pt\rho E}{\pt t}+\dfr{\pt(\rho E+p)u}{\pt x}=0,
\end{equation}
with the total pressure defined as $p \coloneq \sum_{k=1}^K p_k$.
Denote the admissible domain $D \coloneq \R^+ \times \R \times \cdots \times \R$ and the state vector $\bU \coloneq \left[ \rho, \rho\bu, \rho e_1, \ldots, \rho e_K \right]^T \in D$ the mathematical entropy, following~\cite{hantke2025baer}, is defined as
\begin{equation}\label{eq:math-total-entropy}
        s(\bU) \coloneq  \rho \left( E - \sum_{k=1}^K s_k \right) 
        = \rho\left(
          \frac{\|\bu\|_2^2}{2} + \sum_{k=1}^K e_k - \sum_{k=1}^K\frac{\ln(\rho e_k) - \gamma_k \ln{\rho}}{\gamma_k-1}
        \right).
\end{equation}
Then,  the entropy variables read  
{\begin{align}\label{eq:math-entropy-variables}
\begin{split}
    \bW(\bU)&:=\nabla s(\bU) \\
    &= 
    \left[  -\frac{\|\bu\|_2^2}{2} + \sum_{k=1}^K \frac{\gamma_k (1+\ln{\rho}) - \ln(\rho e_k)}{\gamma_k-1} ,
    \bu,
    1-\frac{\rho}{p_1}, \ldots, 1-\frac{\rho}{p_K}
    \right]^T,
    \end{split}
\end{align}}
and their Jacobian matrix, i.e. the Hessian of $s$: 
\begin{equation}\label{eq:math-entropy-hessian}
\bH \coloneq J\bW(\bU) 
= \frac{1}{\rho}
\begin{bmatrix}
 \|\bu\|_2^2 +\sum_{k=1}^K\frac{\gamma_k}{\gamma_k-1} \,\,  & - \bu^T & \frac{-1}{ (\gamma_1 - 1)e_1} & \dots & \frac{-1}{ (\gamma_K - 1)e_K} \\[10pt]
- \bu & \bI_d & \mathbf{0} & \dots &  \mathbf{0} \\[10pt]
\frac{-1}{ (\gamma_1 - 1)e_1} & \mathbf{0} & \frac{(\gamma_1 - 1)\rho^2}{ p_1^2} & \ddots &\vdots \\[10pt]
\vdots & \vdots & \ddots & \ddots & 0\\[10pt]
\frac{-1}{ (\gamma_K - 1)e_K} & \mathbf{0} & \dots & 0 & \frac{(\gamma_K - 1)\rho^2}{ p_K^2}
\end{bmatrix}.
\end{equation}
\begin{lmm}
     The mathematical entropy $s$ defined in~\eqref{eq:math-total-entropy} is strictly convex.
\end{lmm}
\begin{proof}
We examine the Schur complement of the lower-right $K \times K$ block $\bC$ in $\bH$, i.e. $\bC = \frac{1}{\rho} \text{diag}\left( \frac{(\gamma_1 - 1)\rho^2}{ p_1^2}, \dots, \frac{(\gamma_K - 1)\rho^2}{ p_K^2} \right)$ that is clearly positive definite. Let $\bB \in \mathbb{R}^{(d+1) \times K}$ be the corresponding upper-right block, i.e.
\begin{equation}
\bB = \frac{1}{\rho} \begin{bmatrix}
    \frac{-1}{ (\gamma_1 - 1)e_1} & \dots & \frac{-1}{ (\gamma_K - 1)e_K}\\
    \mathbf{0} & \dots & \mathbf{0}
\end{bmatrix},
\end{equation}
and $\bA \in \mathbb{R}^{(d+1)\times(d+1)}$ the upper-left block of $\bH$.

Then the Schur complement is {$\bS = \bA - \bB \bC^{-1} \bB^T$} where
\begin{align*}
{  \bB \bC^{-1} \bB^T}&{= 
                      \frac{1}{\rho^3}
                      \begin{bmatrix}
                        \frac{-1}{ (\gamma_1 - 1)e_1} & \dots & \frac{-1}{ (\gamma_K - 1)e_K}\\
                        \mathbf{0} & \dots & \mathbf{0}
                      \end{bmatrix}
                      \begin{bmatrix}
                        \frac{ p_1^2}{(\gamma_1 - 1)\rho^2} & &\\
                                                    & \ddots & \\
                                                    & & \frac{ p_K^2}{(\gamma_K - 1)\rho^2} 
                      \end{bmatrix}
                      \begin{bmatrix}
                        \frac{-1}{ (\gamma_1 - 1)e_1} & \mathbf{0}\\
                        \vdots & \vdots \\
                        \frac{-1}{ (\gamma_K - 1)e_K} &\mathbf{0}
                      \end{bmatrix}}\\
                    &=\begin{bmatrix}
                      \sum_{k=1}^K \frac{1}{(\gamma_k -1)\rho} & \mathbf{0}\\
                      \mathbf{0}& \mathbf{0}.
                    \end{bmatrix}
\end{align*}
Thus, 
\begin{equation*}{
    \bS = \begin{bmatrix}
        \dfrac{\|\bu\|_2^2}{\rho} +\sum_{k=1}^K \dfrac{\gamma_k - 1}{(\gamma_k-1)\rho} & -\bu^T \\
        -\bu & \dfrac{1}{\rho} \bI_d
    \end{bmatrix}
    =
    \frac{1}{\rho}
    \begin{bmatrix}
        \|\bu\|_2^2 + K  & -\bu^T \\
        -\bu &  \bI_d.
    \end{bmatrix}}
\end{equation*}
The determinant of $\bS $ is obtained exploiting its block structure
\begin{align*}
    \det(\bS) = \frac{1}{\rho}\det(\bI_d)\det(\|\bu\|_2^2 + K - \bu^T \bI_d^{-1}\bu) = \frac{K}{\rho} > 0,
\end{align*}
hence, the Schur complement is positive definite and $s$ is strictly convex.

\end{proof}

As shown in~\cite{chalons2010time}, {the physical entropies $s_k$ satisfy $\frac{\partial s_k}{\partial t} + \bu\cdot\nabla s_k= 0$}. Since the total energy $\rho E$ is conserved, cf.~\eqref{eq:consv-energy}, contracting the multi-temperature Euler system~\eqref{eq:multi-temp-nD} with the entropy variables~\eqref{eq:math-entropy-variables} reveals the entropy conservation law for smooth solutions:
\begin{equation}\label{eq:entropy-cons-law}
    \frac{\partial s}{\partial t} + \text{div}\left((s + p)\bu\right) = 0.
\end{equation}

The eigenvectors are obtained, similarly to~\cite{hantke2025baer}, by firstly switching to primitive variables $ \left[\rho, \bu, p_1,\ldots, p_K\right]^T$:
\begin{subequations}\label{eq:multi-temp-nD-primitive}
\begin{align}
\label{eq:multi-temp-nD-primitive-rho}
&\dfr{\partial \rho}{\partial t} + \operatorname{div}(\rho\bu) = 0,\\ 
\label{eq:multi-temp-nD-primitive-u}
&\dfr{\partial \bu}{\partial t} +  (\bu\cdot\nabla)\bu+\frac{1}{\rho}\nabla p  = \mathbf{0},\\
\label{eq:multi-temp-nD-primitive-e_k}
&\dfr{\partial p_k}{\partial t} + \bu\cdot\nabla p_k+\rho c_k^2\operatorname{div}\bu  = 0, \quad k=1,\ldots,K.
\end{align}
\end{subequations}
Secondly, the system is projected to a unit direction $\bn^T \in \mathbb{R}^d$ such that the partial derivatives transform to  directional derivatives in $\xi \coloneq  \bn \cdot \bx$, i.e.~we recover a quasilinear system $\dfrac{\partial\bV}{\partial t} + \bB(\bV, \bn) \partial_{\xi}\bV = 0 $. In order to compute the Jacobian $\bB$, let $\bc^2 \coloneq \left[c_1^2,  \ldots, c_K^2\right]^T$ and $u_n \coloneq \bn^T\bu$. Moreover, let the  matrix $\bI_d \in\mathbb{R}^{d\times d}$ denote the identity, $\mathbf{0}_K = \left[0, \ldots,0 \right]^T \in \mathbb{R}^K$  a vector of  zeros, and $u_i$ the $i$-th component of $\mathbf{u}$. Then, the Jacobian $\bB$ reads
\begin{align*}
    \bB(\bV, \bn) \coloneq& \sum_{i=1}^d n_i\bB_i=\sum_{i=1}^d n_i
    \begin{bmatrix}
            u_i &  \rho \be_{d, i}^T & \mathbf{0}_K^T\\
            \mathbf{0}_{d} & u_i \bI_d &\frac{1}{\rho}\be_{d, i}\cdot \mathbf{1}_K^T\\
            \mathbf{0}_{K} & \quad  \rho \bc^2 \cdot \be_{d, i}^T & u_i \bI_K
    \end{bmatrix}
    = \begin{bmatrix}
        u_n & \rho \bn^T & \mathbf{0}^T_K\\
        \mathbf{0}_{d} & u_n \bI_d & \frac{1}{\rho} \bn \cdot \mathbf{1}_K^T\\
         \mathbf{0}_{K} & \quad  \rho \bc^2 \cdot \bn^T & u_n \bI_K.
    \end{bmatrix},
\end{align*}
where $\be_{d, i} \in \mathbb{R}^d$ denotes the $i$-th unit vector.
The  eigenvalues of $\bB(\bV, \bn)$ are 
 \begin{equation}\label{eq:multi-temp-nD-eigenval}
      \lambda_{\pm}= u_n \pm c, \quad \lambda_j = u_n, \, j=1, \ldots, d+ K - 1.
 \end{equation}
with the corresponding right eigenvectors
\begin{equation}\label{eq:multi-temp-nD-r-eigenvec}
    \br_{\pm} \coloneq \begin{bmatrix}
        1 \\ 
        \pm \frac{c}{\rho} \bn \\
        \bc^2
    \end{bmatrix}, \quad
    \br_0 \coloneq 
    \begin{bmatrix}
        1 \\ \mathbf{0}_{d} \\ {\mathbf{c}^2 - \frac{c^2}{K} \mathbf{1}_K}
    \end{bmatrix}, \quad
    \br_i \coloneq
    \begin{bmatrix}
        0 \\ 
        \bt_i \\
        \mathbf{0}_K
    \end{bmatrix},    \quad
    \br_k \coloneq
    \begin{bmatrix}
        \mathbf{0}_{d+1}\\
        \sum_{j=1}^k \be_{K, j} - k \be_{K, k+1}
    \end{bmatrix},
\end{equation}
for $i = {1, \ldots, d-1}$ and $k={1, \ldots, K-1}$ and an orthonormal basis $\{\bn, \bt_1, \ldots, \bt_{d-1} \}$ with $\bt_i^T \bt_j = \delta_{ij}$  and $\bt_i ^T \bn =0$ for all $i$.  Thus, the system~\eqref{eq:multi-temp-nD-primitive} written in primitive variables is hyperbolic for any direction $\bn$.

In order to obtain the eigenvectors for the original system~\eqref{eq:multi-temp-nD} in conservative variables, we apply the same method as in~\cite{hantke2025baer}: define $\bJ(\bU) \coloneq \dfrac{\partial \bV}{\partial \bU}(\bU)$ and transform the system~\eqref{eq:multi-temp-nD} as
\begin{equation*}
    \frac{\partial\bU}{\partial t} + \bJ(\bU)^{-1} \cdot \bB(\bV(\bU), \bn) \cdot \bJ(\bU) \bU_{\xi} = \frac{\partial\bU}{\partial t} +  \bA(\bU, \bn)  \bU_{\xi} = 0,
\end{equation*}
where $\bB(\bV, \bn)$ is the projected Jacobian of the fluxes for the system in primitive variables. Thus, the eigenvectors of $\bA(\bU, \bn)$ are
\begin{equation*}{
    \bR \coloneq \bJ(\bU)^{-1}  \left[ \br_-, \br_0, \br_i, \br_k, \br_+ \right],  \quad 
    \bL \coloneq \cdot \left[ \bl_-, \bl_0, \bl_i, \bl_k, \bl_+ \right]\cdot \bJ(\bU).}
\end{equation*}
and the left eigenvectors $\bl$ and the matrices $\bJ$ and $\bJ^{-1}$ are provided in Appendix~\ref{app:eigen}.
By the transformation to conservative variables by means of $\bJ$ and $\bJ^{-1}$ one recovers the eigenvectors of the original system~\eqref{eq:multi-temp-nD} which is, consequently, hyperbolic as well as the system in primitive variables~\eqref{eq:multi-temp-nD-primitive}.

\vspace{2mm}
\section{Hugoniot relations and resolving shocks}\label{sec:wave-right}
We consider the system~\eqref{eq:multi-temp-nD}  projected in the direction $\xi =  \bn \cdot \bx$ where $\bn$ is a unit normal vector of an interface  and define the normal velocity $u_n \coloneq \bn^T \bu$:
\begin{subequations}\label{eq:multi-temp-projected-nD}
\begin{align}
&\dfr{\partial \rho}{\partial t} + \partial_{\xi} \rho \bu = 0,\label{eq:multi-temp-nD-projected-rho}\\ 
&\dfr{\partial \rho \bu}{\partial t} + \partial_{\xi}\left( \rho u_n \bu + p\bn  \right) = \mathbf{0},\label{eq:multi-temp-nD-projected-rho-u}\\
&\dfr{\partial\rho e_k}{\partial t} + \partial_{\xi}(\rho e_k\bu) + p_k \partial_{\xi} \bu = 0, \quad k=1,\ldots,K,\label{eq:multi-temp-nD-projected-rho-e_k}
\end{align}
\end{subequations}
Due to the lack of a \emph{single} Rankine-Hugoniot condition, the difficulty of solving Riemann problems for the multi-temperature Euler equations~\eqref{eq:multi-temp-projected-nD} lies in resolving shocks, that is the topic of this section.

We neglect the tangential direction, in which the velocity $\bu$  across a shock is continuous. Then, the Rankine-Hugoniot relations of~\eqref{eq:multi-temp-nD-projected-rho}, \eqref{eq:multi-temp-nD-projected-rho-u} in normal direction $\bn$ together with the total energy conservation~\eqref{eq:consv-energy} read
\begin{equation}\label{eq:RH-Euler}
\begin{split}
&\rho (u_n-\sigma) = \rho_- (u_{n,-}-\sigma), \ \ \ \ 
\rho (u_n-\sigma)^2+p = \rho_- (u_{n,-}-\sigma)^2 + p_-,\\
&\left[\dfr 12\rho (u_n-\sigma)^2+e+p\right](u_n-\sigma) = \left[\dfr 12\rho_- (u_{n,-}-\sigma)^2 + e_- + p_-\right](u_{n,-}-\sigma),
\end{split}
\end{equation}
where  $e=\sum_{k=1}^K e_k$, $p_-= \sum_{k=1}^K p_{k,-}$ and $e_-= \sum_{k=1}^K e_{k,-}$. 
The minus sign at the subscripts stands for a pre-shock state.
The primitive unknown variables of the post-shock state are denoted by $\bV_n = (\rho, u_n, p_1, \ldots p_K)$ and  are connected to the pre-wave state $\bV_{n,-}\coloneq (\rho_-, u_{n,-}, p_{1,-}, \ldots p_{K,-})$ by a shock wave that travels at  speed $\sigma$. Note, that the internal energies $e_k$  depend on $\tau$.
The first two equations in \eqref{eq:RH-Euler} combine to
\begin{equation}\label{eq:p-u-shock-0}
u_n=u_{n,-} \mp \sqrt{(p-p_-)(\tau_--\tau)},
\end{equation}
where the minus sign applies for left-going shock waves, i.e. corresponding to the eigenvalue $\lambda_-$, and vice versa. Together with the last two equations they lead to the Hugoniot relation
\begin{equation}\label{eq:Hugoniot-total}
e - e_- + \dfr{p + p_-}{2}(\tau-\tau_-)=0.
\end{equation}
For the single-phase Euler equations, the Hugoniot relation \eqref{eq:Hugoniot-total}, together with the EOS $e=e(p,\tau)$, determine a curve in the $(p,\tau)$ plane, which is called the Hugoniot curve~\cite{hugoniot-1887,hugoniot-1889}. This Hugoniot curve determines the shock wave uniquely, i.e. its image consists of all states $\bV_n$ that can be connected to $\bV_{n,-}$. 
In case of the multi-temperature Euler equations \eqref{eq:multi-temp-nD}, however,
the Hugoniot relation \eqref{eq:Hugoniot-total} together with \eqref{eq:eos} yield merely a \emph{Hugoniot hypersurface} in the phase space $\{(p_1, \ldots,p_K,\tau) \in \mathbb{R}^{K+1}\}$.
Thus, $K-1$ more conditions in addition to the Hugoniot relation associated with total energy are required. These conditions are  lost during the process of model reduction to a single conservation equation for the total energy.
Alternatively, phasic Hugoniot relations that hold for each $k \in \{1, \ldots, K\}$  instead of the one associated with the total energy, could be considered. 

At the same time, the DLM-theory~\cite{dal1995}, which defines the weak formulation for non-conservative products, 
requires shock waves to fulfill the \emph{generalized} Rankine-Hugoniot conditions. 
For a quasilinear system of hyperbolic PDEs,
\begin{equation*}
\dfr{\pt\bU}{\pt t}+\bA(\bU)\dfr{\pt\bU}{\pt x}=0,
\end{equation*}
the generalized Rankine-Hugoniot relation reads
\begin{equation}\label{eq:RH-general}
\sigma (\bU-\bU_-)=\d\int_0^1\bA(\Phi(\eta;\bU_-,\bU))\dfr{\pt\Phi(\eta;\bU_-,\bU)}{\pt\eta}d\eta,
\end{equation}
where $\Phi \,\colon\, [0,1] \times D \times D \to D$  is a path in the phase space connecting $\bU$ and $\bU_-$ such that 
\begin{align*}
\begin{split}
&\Phi(0, \bU_-,\bU) = \bU_-, \qquad \Phi(1,\bU_-,\bU) = \bU,\\
&\Phi(\eta, \bU,\bU) = \bU,\\
&\Phi(\eta,\bU_-,\bU) = \Phi(1-\eta, \bU,\bU_-),
\end{split}
\end{align*}
for all $ \bU_-,\bU \in D$ and $ \eta \in [0,1]$. In case $\bA(\bU)$ is not a Jacobian matrix, i.e.~the quasilinear system cannot be transformed to a conservation law, the integral in~\eqref{eq:RH-general} depends on the choice of the path $\Phi$.
{Following the DLM framework~\cite{dal1995}, we additionally assume that $\Phi$ is locally Lipschitz continuous in its arguments. More precisely, for every bounded set $B\Subset D$ there is a constant $L_B>0$ such that $|\partial_\eta\Phi(\eta;\bU_L,\bU_R)|\leq L_B|\bU_R-\bU_L|$ for almost every $\eta\in[0,1]$ and all $\bU_L,\bU_R\in B$. This regularity ensures that the non-conservative product defines a Borel measure across a jump.}

One immediately discovers that the generalized Rankine-Hugoniot conditions only push the question of the uniqueness one step forward, instead of solving it. It is well-known that the ambiguity  has to be resolved by selecting a path $\Phi$ based on the underlying physics. As a matter of fact, it seems that the Hugoniot relation is more fundamental to study shocks.

 In the following subsections we will  analyze two different shock solutions and show that both fulfill a set of conditions necessary for a valid model. In the subsequent sections we will show the relation between these shock paths and the according numerical schemes.


\vspace{2mm}
\subsection{ Hugoniot relation for the segment path}\label{sec:segment-path-hugoniot}
We consider the segment path $\Phi$ defined as
\begin{equation*}
\begin{split}
&u_n(\Phi(\eta))=(1-\eta)u_{n,-}+\eta u_n, \\
&p_k(\Phi(\eta))=(1-\eta)p_{k,-}+\eta p_k, \qquad\quad k =1, \ldots, K, \\
&(\rho e_k)(\Phi(\eta))=(1-\eta)\rho_-e_{k,-}+\eta\rho e_k,
\end{split}
\end{equation*}
that connects the pre- and post-shock states by a straight line. The generalized Rankine-Hugoniot condition for the $k$-th internal energy equation~\eqref{eq:multi-temp-nD-projected-rho-e_k} yields
\begin{equation*}\label{eq:RH-ion-straight}
\rho v_n e_k - \rho_- v_{n,-} e_{k,-} + \dfr{p_k + p_{k,-}}{2}(v_n-v_{n,-})=0,
\end{equation*}
where $v_n=u_n-\sigma$ and $\sigma$ is the shock speed.
By the fact that $\rho v_n = \rho_- v_{n,-}$ the phasic Hugoniot relation
\begin{equation}\label{eq:k-th-Hugoniot-straight}\tag{$\text{H}^{\text{seg}}$}
e_k - e_{k,-} + \dfr{p_k + p_{k,-}}{2}(\tau-\tau_-)=0,
\end{equation}
holds where $\tau = \frac{1}{\rho}$ is the specific volume. 
{The phasic Hugoniot relations~\eqref{eq:k-th-Hugoniot-straight} were studied by Prokopov~\cite{prokopov2004riemann} for the Riemann problem in three-temperature gas dynamics, together with the observation that their sum conserves total energy. To our knowledge, their Courant--Friedrichs thermodynamic properties have not previously been developed in detail. We provide that analysis below, including second-order contact with the isentropic curve and entropy production on the admissible compressive branch.}

\vspace{2mm}
\paragraph{\textit{Energy conservation.}} Summing equations~\eqref{eq:k-th-Hugoniot-straight} over $k$ yields exactly  \eqref{eq:Hugoniot-total}, which is the Hugoniot relation of the Euler equations.


\vspace{2mm}
\paragraph{\textit{Contact of the Hugoniot curve with the isentropic curve.}} Differentiating the Hugoniot relation~\eqref{eq:k-th-Hugoniot-straight} one recovers
\begin{equation*}
2de_k  = (\tau_--\tau)dp_k + \big(p_{k,-}+p_k\big)d\tau.
\end{equation*}
According to the Gibbs relation \eqref{eq:gibbs},
\begin{equation}\label{eq:ds-1-straight}
2T_kds_k = (\tau_--\tau)dp_k - \big(p_{k,-}-p_k\big)d\tau,
\end{equation}
which indicates that
\beql{eq:ds-sp-1st}
\left.\left(T_kds_k\right)\right|_{\tau=\tau_-}=0.
\eeq
Differentiate \eqref{eq:ds-1-straight} to obtain
\begin{equation}\label{eq:ds-2-straight}
2d(T_kds_k) = (\tau_--\tau)d^2p_k.
\end{equation}
The left-hand side is $2d(T_kds_k)=2dT_kds_k+2T_kd^2s_k$ while the right-hand side is $(\tau_--\tau)d^2p_k$, that implies
\beql{eq:ds-sp-2nd}
\left.\left(T_kd^2s_k\right)\right|_{\tau=\tau_-}=0.
\eeq

Next, denote the isentropic trajectory by $(\mathfrak{p}_1(\tau), \ldots, (\mathfrak{p}_K(\tau))$. For polytropic gases its components are
\beql{eq:isentropic-curve-poly}
\mathfrak{p}_k(\tau)\coloneq p_{k,-}\left(\dfr{\tau_-}{\tau}\right)^{\gm_k}, \quad k =1,\ldots, K,
\eeq
Hence, the differences \eqref{eq:ds-sp-1st} and \eqref{eq:ds-sp-2nd} yield
\beqs
\left.\dfr{dp_k}{d\tau}\right|_{\tau=\tau_-}=\left.\dfr{d\mathfrak{p}_k}{d\tau}\right|_{\tau=\tau_-}, \ \ \ \ \left.\dfr{d^2p_k}{d\tau^2}\right|_{\tau=\tau_-}=\left.\dfr{d^2\mathfrak{p}_k}{d\tau^2}\right|_{\tau=\tau_-}.
\eeqs
The derivative of~\eqref{eq:ds-2-straight} yields
\begin{equation}\label{eq:ds-3-straight}
2d^3(T_ids_k) = (\tau_--\tau)d^3p_k-d\tau d^2p_k
\end{equation}
and by taking the limit $\tau\rightarrow\tau_-$, \eqref{eq:ds-3-straight} we conclude that
\begin{equation}\label{eq:ds-3-straight-simplified}
T_k\dfr{d^3s_k}{d^3\tau}=-\dfr{1}{2}\dfr{d^2p_k}{d\tau^2},
\end{equation}
at $\tau=\tau_-$. Therefore, the thermodynamical hypothesis $\frac{\pt^2p_k}{\pt\tau^2}>0$ implies $\frac{d^3s_k}{d^3\tau}<0$ at $\tau=\tau_-$. 

\vspace{2mm}
\paragraph{\textit{Entropy production}} We show that  for all $\tau<\tau_-$,
\beqs
\dfr{ds_k}{d\tau}<0, \quad k =1,\ldots, K.
\eeqs
Since the total Hugoniot relation~\eqref{eq:Hugoniot-total} can be split into a sum of phasic Hugoniots~\eqref{eq:k-th-Hugoniot-straight} we apply the classical proof in~\cite[\S 65]{Courant-Friedrichs} to the projections of the Hugoniot curve on the $(p_k,\tau)$-plane.
Thereby, one can show that $ds_k$ does not vanish along the Hugoniot curve for any $k=1,\ldots, K$. 
By the fact that $\frac{d^3s_k}{d^3\tau}({\tau_-})<0$, cf.~\eqref{eq:ds-3-straight-simplified}, $\frac{ds_k}{d\tau}$ is always negative along the Hugoniot curve except the initial point $(p_{1,-}, \ldots, p_{K,-}, \tau_-)$.

\subsection{Vanishing viscosity Hugoniot}\label{sec:vanishing-visc-hugoniot}
In~\cite{chalons2005riemann} the shock solution of \eqref{eq:multi-temp-nD} is defined by means of the vanishing viscosity limit of the traveling wave solution to the Navier-Stokes equations. Setting $m\coloneq\rho_-(u_{n,-}-\sigma)$ the internal energies of the traveling wave solution satisfy
{
\begin{equation}\label{eq:ode-ion-vv}
\dfr{de_k}{d\tau} + p_k = \dfr{\mu_k}{\mu}\mathcal{F}(\tau,m^2), \text{ for } k=1,\dots,K,
\end{equation}}
with $\mu=\sum_{k=1}^K\mu_k$. Integrating \eqref{eq:ode-ion-vv} from $\tau_-$ to $\tau$, for $k=1,\dots,K$, we obtain the curves
{
\begin{equation}\label{eq:integral-curve-vv}
\begin{split}
e_k(\tau)-e_{k,-} = \dfr{\mu_k}{\mu}&\left[\dfr{m^2}{2}(\tau-\tau_-)^2-p_-(\tau-\tau_-)\right]\\
&+\left(\dfr{\mu_k}{\mu}-1 \right)\d\int_{\tau_-}^\tau \pi^\text{vis}_k(\om;m^2)d\om + \dfr{\mu_k}{\mu}\d\sum_{j\neq k}\int_{\tau_-}^\tau \pi^\text{vis}_j(\om;m^2)d\om.
\end{split}
\end{equation}
}

However, an arbitrary {point} $(\tau, e_1,\dots,e_K)$ on the integral curve of the ordinary differential equation (ODE) system~\eqref{eq:ode-ion-vv} is not necessarily an admissible post-shock state. 
In order to find admissible states, define
\begin{equation}\label{eq:RH-mom}
\mathcal{F}(\tau,m^2;p_{1,-},\dots,p_{K,-})\coloneq m^2(\tau-\tau_-)+\sum_{k=1}^{K}\left[\pi^\text{vis}_k(\tau;m^2)-p_{k,-}\right],
\end{equation}
where $\pi^\text{vis}_k(\tau;m^2)$ is the dependence of $p_k$ on $\tau$ for a given $m$ along the integral curve of the ODEs \eqref{eq:ode-ion-vv}, see~\cite{chalons2005riemann}. The map $\mathcal{F}$ is defined such that the requirement of  $\mathcal{F}=0$ is equivalent to the Rankine--Hugoniot relation for the momentum equation \eqref{eq:multi-temp-nD-projected-rho-u}. Thus, for any fixed $m$, we require $\tau$ to fulfill $\mathcal{F}=0$. 
With these tools one derives the Hugoniot relation of the vanishing viscosity path.

The first step to determine the Hugoniot curve under the vanishing viscosity assumption is to find  valid $\tau=\tau(m)$ (not necessarily given as an explicit expression) such that $\mathcal{F}=0$.
The second step is to reverse the dependence between $\tau$ and $m$ by using
\begin{equation}\label{eq:def-m2}
m^2(\tau)=\dfr{p-p_-}{\tau_--\tau}.
\end{equation}
Thus, we substitute \eqref{eq:def-m2} in \eqref{eq:integral-curve-vv}, eliminating $m$, and obtain the curves
{\begin{equation}\label{eq:Hugoniot-vv}\tag{$\text{H}^{\text{vis}}$}
\begin{split}
e_k(\tau)-e_{k,-} &+ \dfr{\mu_k}{\mu}\dfr{p_-+p}{2}(\tau-\tau_-)
\\
&+\left(1-\dfr{\mu_k}{\mu}\right)
\d\int_{\tau_-}^\tau \pi^\text{vis}_k(\om;m^2(\tau))d\om - \dfr{\mu_k}{\mu}\d\sum_{j\neq k}\int_{\tau_-}^\tau \pi^\text{vis}_j(\om;m^2(\tau))d\om = 0,
\end{split}
\end{equation}}
that are the Hugoniot relations obtained by the vanishing viscosity approach for species $k$, with $1\leq k\leq K$.

\begin{rmrk}\label{rmrk:pi_k}
 Note that the notation $\pi^\text{vis}_k$ is used here to emphasize the dependence of $p_k$ on $\tau$ along the traveling wave solution that is, however, not constrained by the thermodynamics defined by the EOS. For the terms in the square bracket on the right-hand side of \eqref{eq:integral-curve-vv}, the result is independent of the integral trajectory since these are primitive functions. 
 \\
 Only after eliminating the explicit dependence on $m$ in conjunction with the condition $\mathcal{F}(\tau,m^2)=0$, one obtains the the sought Hugoniot curve~\eqref{eq:Hugoniot-vv}. The above computation shows that the Hugoniot curve~\eqref{eq:Hugoniot-vv} corresponds to the family of stagnation points of $\mathcal{F}$ for varying $m^2$, i.e. points $(\tau,m^{2, *})$ such that $\mathcal{F}(\tau,m^{2, *})=0$ (see~\cite{chalons2005riemann}).
\end{rmrk}

The work of Chalons and Coquel~\cite{chalons2005riemann} establishes the Riemann solution framework associated with the vanishing viscosity Hugoniot relation and the corresponding solver construction, as stated above. In contrast, the present study focuses on the thermodynamic structure of the vanishing viscosity shock solution, rather than on solving the Riemann problem alone. Specifically, the following analysis characterizes its thermodynamic consistency through demonstrating the energy Hugoniot relation, contact with isentropic curves, and entropy behavior along admissible branches.
This complements the contribution of the present work by providing an explicit analytical characterization of physically relevant shock states for the multi-temperature Euler model.

\vspace{2mm}
\paragraph{\textit{Energy conservation}}
Since by assumption $\sum_{k=1}^K\frac{\mu_k}{\mu}=1$ holds, the sum of the two Hugoniot relations in \eqref{eq:Hugoniot-vv} immediately yields
\begin{equation*}
e - e_-+\dfr{p_-+p}{2}(\tau-\tau_-)=0.
\end{equation*}
This is exactly the total Hugoniot relation~\eqref{eq:Hugoniot-total}.

\vspace{2mm}
\paragraph{\textit{Contact of the Hugoniot curve with the isentropic curve}}  In this part we follow the same procedure as applied in~\eqref{eq:ds-1-straight}-\eqref{eq:ds-3-straight-simplified}. Differentiating \eqref{eq:Hugoniot-vv} and using the Gibbs relation \eqref{eq:gibbs} we get
\begin{equation}\label{eq:ds-1-vv}
2T_kds_k = \dfr{\mu_k}{\mu}\Big[(\tau_--\tau)dp + (p - p_-)d\tau\Big]+\Omega_0d\tau, \, \, \, k=1,\dots,K,
\end{equation}
where
\beqs
\Omega_0=
\left(2-\dfr{2\mu_k}{\mu}\right)
\d\int_{\tau_-}^\tau 
\dfr{\pt\pi^\text{vis}_k}{\pt(m^2)}(\om;m^2(\tau))\dfr{dm^2}{d\tau}(\tau)
d\om
  - \dfr{2\mu_i}{\mu}\d\sum_{j\neq k}\int_{\tau_-}^\tau
\dfr{\pt\pi^\text{vis}_j}{\pt(m^2)}(\om;m^2(\tau))\dfr{dm^2}{d\tau}(\tau)
d\om.
\eeqs
Further differentiating \eqref{eq:ds-1-vv} leads to
\begin{equation}\label{eq:ds-2-vv}
2d(T_kds_k) = \dfr{\mu_k}{\mu}(\tau_--\tau)d^2p+(\Omega_1+\Omega_2)d\tau^2,
\end{equation}
with
{
\beqs
\begin{split}
\Omega_1 = &
\left(2-\dfr{2\mu_k}{\mu}\right)
\dfr{\pt\pi^\text{vis}_k}{\pt(m^2)}(\tau;m^2(\tau))\dfr{dm^2}{d\tau}(\tau)
- \dfr{2\mu_k}{\mu}\sum_{j\neq k}\dfr{\pt\pi^\text{vis}_j}{\pt(m^2)}(\tau;m^2(\tau))\dfr{dm^2}{d\tau}(\tau),
\\
\Omega_2 = &
\left(2-\dfr{2\mu_k}{\mu}\right)
\d\int_{\tau_-}^\tau 
\left[
\dfr{\pt^2\pi^\text{vis}_k}{[\pt(m^2)]^2}(\om;m^2(\tau))\left(\dfr{dm^2}{d\tau}(\tau)\right)^2
+\dfr{\pt\pi^\text{vis}_k}{\pt(m^2)}(\om;m^2(\tau))\dfr{d^2m^2}{d\tau^2}(\tau)
\right]d\om
\\
&
  - \dfr{2\mu_k}{\mu}\d\sum_{j\neq k}\int_{\tau_-}^\tau
\left[
\dfr{\pt^2\pi^\text{vis}_j}{[\pt(m^2)]^2}(\om;m^2(\tau))\left(\dfr{dm^2}{d\tau}(\tau)\right)^2
+\dfr{\pt\pi^\text{vis}_j}{\pt(m^2)}(\om;m^2(\tau))\dfr{d^2m^2}{d\tau^2}(\tau)
\right]d\om.
\end{split}
\eeqs}
Since $\Omega_0\rightarrow0$, $\Omega_1\rightarrow0$ and $\Omega_2\rightarrow0$ as $\tau\rightarrow\tau_-$, see Appendix \ref{app:regularity}, 
\beqs
\left.\left(T_kds_k\right)\right|_{\tau\rightarrow\tau_-}=0,
\ \ \ 
\left.\left(T_kd^2s_k\right)\right|_{\tau\rightarrow\tau_-}=0.
\eeqs

Along the Hugoniot curve, we  have, by \eqref{eq:ds-1-vv} and \eqref{eq:ds-2-vv},
\beqs
\left.\dfr{dp_k}{d\tau}\right|_{\tau=\tau_-}=\left.\dfr{d\mathfrak{p}_k}{d\tau}\right|_{\tau=\tau_-}, \ \ \ \ \left.\dfr{d^2p_k}{d\tau^2}\right|_{\tau=\tau_-}=\left.\dfr{d^2\mathfrak{p}_i}{d\tau^2}\right|_{\tau=\tau_-}.
\eeqs
which indicates the vanishing viscosity Hugoniot curve contacts with the isentropic curve up to the second order.

\vspace{2mm}
\paragraph{\textit{Entropy production} } The entropy production property of the vanishing viscosity shock solution is stated as the following theorem.
\begin{thrm}\label{thm:entropy-vv}
Assume $\gamma_k>1$, $p_{k,-}>0$, and $\mu_k>0$ for $k=1,\ldots,K$. Then there exists $\epsilon^*\in(0,1)$ such that
\begin{equation}\label{eq:entropy-production}
\frac{ds_k}{d\tau}<0,
\qquad k=1,\ldots,K,
\end{equation}
for every $\tau\in(\epsilon^*\tau_-,\tau_-)$. Thus each phasic physical entropy increases across a sufficiently weak compressive shock along the Hugoniot curve \eqref{eq:Hugoniot-vv}.
\end{thrm}

{
Following \cite{chalons2005riemann}, the post-shock pressures and the specific volume of polytropic gases satisfy the ODE
\begin{equation}\label{eq:ode-vv}
\left\{
\bga{l}
\dfr{d{\bf\Theta}}{dx}-\mathcal{M}{\bf\Theta}=\alpha\ba + (1-\exp(-x))\bb,\\[4mm]
{\bf\Theta}(0)=\mathbf{0},
\eda
\right.
\end{equation}
where $x=-\ln(\frac{\tau}{\tau_-})$, $\alpha=\frac{({\rho_-}{c_-})^2}{m^2}$. The entries of $\mathcal{M}$ is $\mathcal{M}_{kj} = \gamma_k\delta_{kj} - \frac{\mu_k}{\mu}(\gamma_j-1)$. Therefore,
\begin{equation*}
\mathcal{M} = \mathbf{diag}(\gamma_1,\dots,\gamma_K) - \left[\frac{\mu_1}{\mu} \, \dots \, \frac{\mu_K}{\mu}\right]^\top[(\gm_1-1) \, \dots \, (\gamma_K-1)],
\end{equation*}
where $\mu=\sum_{k=1}^K\mu_k$.
The constant vectors of the right-hand side of~\eqref{eq:ode-vv} are
\beql{eq:def-vv-a-and-b}
\ba = \left[
\bga{c}
\dfr{\gm_1}{\gm_1-1}\dfr{p_{1,-}\tau_-}{{c_-}^2}\\[2mm]
\vdots\\[2mm]
\dfr{\gm_K}{\gm_K-1}\dfr{p_{K,-}\tau_-}{{c_-}^2}
\eda
\right], \ \ \ \bb=\left[
\bga{c}
\dfr{\mu_1}{\mu}\\[2mm]
\vdots\\[2mm]
\dfr{\mu_K}{\mu}
\eda
\right].
\eeq
The solution ${\bf\Theta}$ in terms of specific volume and pressure is expressed as
\begin{equation}\label{eq:def-Theta}
{\bf\Theta}=
\left[
\bga{c}
\theta_1 \\[3.5mm]
\vdots \\[3.5mm]
\theta_K
\eda
\right]
=\left[
\bga{c}
\dfr{1}{\gm_1-1}\dfr{p_1-p_{1,-}}{m^2\tau_-}\\[2mm]
\vdots \\[2mm]
\dfr{1}{\gm_K-1}\dfr{p_K-p_{K,-}}{m^2\tau_-}
\eda
\right].
\end{equation}
The ODE \eqref{eq:ode-vv} can be analytically solved as
\begin{equation}\label{eq:Theta-tau-vv}
\displaystyle{\bf\Theta}(x)=\alpha\int_0^x{\exp}[\mathcal{M}(x-y)]dy \ \ba + \int_0^x{[1-\exp(-y)]\exp}[\mathcal{M}(x-y)]dy \ \bb.
\end{equation}
Note that an arbitrary pair $(x, {\bf\Theta})$ is not necessarily an admissible post-shock state. 
To find the admissible pair, we recall equation~\eqref{eq:RH-mom}:
\begin{equation*}
\mathcal{F}(\tau,m^2;p_{1,-},\dots,p_{K,-})=m^2(\tau-\tau_-)+\displaystyle\sum_{k=1}^K\left[\pi^\text{vis}_k(\tau;m^2)-p_{k,-}\right],
\end{equation*}
and consider, again, pairs $(x, {\bf\Theta})$ with  $x$ chosen such that $\mathcal{F}=0$.  The proof of Theorem~\ref{thm:entropy-vv} is based on the following five lemmata.}


\vspace{2mm}
{
\begin{lmm}
The Hugoniot curve defined by the stagnation points s.t. $\mathcal{F}=0$, along the integral curve of the ODE \eqref{eq:ode-vv}, follows the \emph{nonlinear} ODE system
\begin{equation}\label{eq:ode-hugoniot-vv-pres}
\left(
\mathbf{I}_K+
\left[
\begin{array}{ccc}
{w}_1-{z}_1  & \dots & {w}_1-{z}_1 \\[5mm]
 & \dots &  \\[5mm]
{w}_K-{z}_K  & \dots & {w}_K-{z}_K
\end{array}
\right]
\right)
\left[
\begin{array}{c}
\dfr{d p_1}{d\tau} \\[3mm]
\vdots \\[3mm]
\dfr{d p_K}{d\tau}
\end{array}
\right]
=
\left[
\begin{array}{c}
-m^2({w}_1-{z}_1) - \dfr{\gamma_1p_1}{\tau} \\[3mm]
\vdots \\[3mm]
-m^2({w}_K-{z}_K) - \dfr{\gamma_Kp_K}{\tau} 
\end{array}
\right],
\end{equation}
where $\mathbf{I}_K$ is an $K\times K$ identity matrix, $m^2$ is defined by \eqref{eq:def-m2},
\begin{equation}\label{eq:coeff-ode-vv-pres}
{w}_{k} \coloneq\dfr{(\gamma_{k}-1)\rho_-{c_-}^2\tilde{a}_{k}}{p-p_-},
 \ \ \ 
{z}_{k}\coloneq\dfr{p_{k}-p_{k,-}}{p-p_-}, \ \ \ k=1,\dots,K,
\end{equation}
and
\begin{equation}\label{eq:def-tilde-a}
\widetilde{\mathbf{a}}=
\left[
\tilde{a}_{1},\,
\dots,\,
\tilde{a}_{K}
\right]^\top
\coloneq
\int_0^x{\exp}[\mathcal{M}(x-y)]dy \ \ba.
\end{equation}
\end{lmm}
\begin{proof}
By the solution \eqref{eq:Theta-tau-vv} and the definition of $x$,
\beql{eq:deriv-theta-x}
\dfr{\pt\theta_k}{\pt((m^2)^{-1})}
=
(\rho_-c_-)^2\tilde{a}_k,
\ \ \ 
\dfr{dx}{d\tau}=-\dfr{1}{\tau}.
\eeq
On the one hand, by regarding $\tau$ as the independent variable and $m^2$ as the parameter in \eqref{eq:ode-vv}, differentiate $\theta_i$ with respect to $\tau$ along the Hugoniot curve to obtain
\begin{equation*}\label{eq:deriv-theta-2-left}
\dfr{d\theta_k}{d\tau}=
-\dfr{1}{\tau}\dfr{\pt\theta_k}{\pt x}
-\dfr{(\rho_-{c_-})^2\tilde{a}_k}{(m^2)^2}\dfr{dm^2}{d\tau}.
\end{equation*}
On the other hand, by the definition \eqref{eq:def-Theta},
\begin{equation*}\label{eq:deriv-theta-2-right}
\dfr{d\theta_k}{d\tau}=
\frac{1}{\gamma_k-1}
\left[
\frac{1}{m^2\tau_-}
\dfr{dp_k}{d\tau}
-
\dfr{p_k-p_{k,-}}{\tau_-}
\frac{1}{(m^2)^2}
\dfr{dm^2}{d\tau}
\right].
\end{equation*}
By combining the above two equations and substituting \eqref{eq:deriv-theta-x}, one gets
\beql{eq:deriv-theta-2}
\dfr{dp_k}{d\tau}=
-(\gamma_k-1)\dfr{\rho_-{c_-}^2\tilde{a}_k}{m^2}\dfr{dm^2}{d\tau}
+
\dfr{p_k-p_{k,-}}{m^2}\dfr{dm^2}{d\tau}
-(\gamma_k-1)m^2\dfr{\tau_-}{\tau}\dfr{\pt\theta_k}{\pt x}.
\eeq
Furthermore, by \eqref{eq:def-m2},
\beql{eq:deriv-m2}
\dfr{dm^2}{d\tau}
=\dfr{1}{\tau_--\tau}\displaystyle\sum_{k=1}^K\dfr{dp_k}{d\tau}+\dfr{p-p_-}{(\tau_--\tau)^2}
=\dfr{1}{\tau_--\tau}\left(\displaystyle\sum_{k=1}^K\dfr{dp_k}{d\tau}+m^2\right).
\eeq
By the ODE \eqref{eq:ode-vv}, and the definition of $\Theta$,
\beql{eq:dtheta-dx}
\begin{split}
\dfr{\pt\theta_k}{\pt x}
&=\sum_{j=1}^K\mathcal{M}_{kj}\theta_j
+\alpha{a_k} + \Big[1-\exp(-x)\Big]b_k
\\
&=\gamma_k\theta_k
-\dfr{\mu_k}{\mu}\Big[
\sum_{j=1}^K(\gamma_j-1)\theta_j
\Big]
+\dfr{\rho_-p_{k,-}}{m^2}\dfr{\gamma_k}{\gamma_k-1}
+
\dfr{\tau_--\tau}{\tau_-}
\dfr{\mu_k}{\mu}\\
&=\dfr{\gamma_k}{\gamma_k-1}\dfr{\tau_--\tau}{\tau_-}\dfr{p_{k}-p_{k,-}}{p-p_-}
-\dfr{\mu_k}{\mu}\dfr{\tau_--\tau}{\tau_-}
\Big[
\sum_{j=1}^K\dfr{p_{j}-p_{j,-}}{p-p_-}
\Big]
+\dfr{\rho_-p_{k,-}}{m^2}\dfr{\gamma_k}{\gamma_k-1}
+
\dfr{\tau_--\tau}{\tau_-}
\dfr{\mu_k}{\mu}\\
&=\dfr{\gamma_k}{\gamma_k-1}\dfr{1}{m^2}\dfr{1}{\tau_-}
p_k.
\end{split}
\eeq
By substituting \eqref{eq:deriv-m2} and \eqref{eq:dtheta-dx} into \eqref{eq:deriv-theta-2}, we obtain
\beql{eq:deriv-theta-4}
\begin{split}
&\dfr{dp_k}{d\tau}
+
\dfr{(\gamma_k-1)\tilde{a}_k\rho_-{c_-}^2 - (p_k-p_{k,-})}{p-p_-}\sum_{k=1}^K\dfr{dp_k}{d\tau}
\\
= &
-\dfr{(\gamma_k-1)\tilde{a}_k\rho_-{c_-}^2 - (p_k-p_{k,-})}{\tau_--\tau}
-\frac{\gamma_kp_k}{\tau}.
\end{split}
\eeq

Recall the definition \eqref{eq:coeff-ode-vv-pres}. We obtain the system of ODEs \eqref{eq:ode-hugoniot-vv-pres}.
\end{proof}}

\vspace{2mm}
{
\begin{lmm}\label{lem:isentropic-condition}
Along the Hugoniot curve, suppose that $w_k\neq z_k$ and $w_j\neq z_j$ for some $\tau$ and $k,j\in\{1,\dots,K\}$. Then $\frac{ds_k}{d\tau}=0$ if and only if $\frac{ds_j}{d\tau}=0$.
\end{lmm}
\begin{proof}
By the ideal gas EOS \eqref{eq:eos} we have
\beql{eq:p-to-s-polytropic}
\dfr{dp_k}{d\tau}=\tau^{-\gm_k}\dfr{ds_i}{d\tau}-\dfr{\gm_kp_k}{\tau}, \ \ \ k=1,\dots,K.
\eeq
If $\frac{ds_k}{d\tau}=0$, \eqref{eq:p-to-s-polytropic} yields $\dfr{dp_k}{d\tau}=-\dfr{\gamma_kp_k}{\tau}.$
Substituting the above result into the first equation of \eqref{eq:ode-hugoniot-vv-pres} reads
\beqs
({w}_k-{z}_k)\sum_{l=1}^{N}\dfr{dp_l}{d\tau}
=
-m^2({w}_k-{z}_k).
\eeqs
Combining with the $j$-th equation in \eqref{eq:ode-hugoniot-vv-pres}, by the fact that $w_k-z_k\neq0$ and  $w_j-z_j\neq0$, one obtains
\beqs
\dfr{dp_j}{d\tau}
=
-\dfr{\gamma_jp_j}{\tau},
\eeqs
which, by \eqref{eq:p-to-s-polytropic}, implies
$\frac{ds_j}{d\tau}=0$.
\end{proof}
}

\vspace{2mm}
{
\begin{lmm}
The local velocity of the integral curve for the ODE system \eqref{eq:ode-hugoniot-vv-pres} is isentropic if and only if one of the following two conditions holds:
\begin{enumerate}[label=\roman*.]
\item\label{item:condition-1} The equation
\beql{eq:condition-isentropic-1}
\displaystyle\sum_{k=1}^K\dfr{\gm_kp_k}{\tau}=\dfr{p-p_-}{\tau_--\tau}
\eeq
is satisfied, and the relation
\beql{eq:condition-isentropic-2}
p_k = p_{k,-}+(\gm_k-1)\rho_-{c_-}^2\tilde{a}_k
\eeq
holds for some (possibly none), but not all, $k \in \{1,\dots,K\}$.
\item\label{item:condition-2} Equation \eqref{eq:condition-isentropic-2} holds for all $k\in\{1,\dots,K\}$.
\end{enumerate}
\end{lmm}
\begin{proof}
Denote the coefficient matrix in \eqref{eq:ode-hugoniot-vv-pres} as $\mathbf{A}=\mathbf{I}_K+\mathbf{B}$. Its inverse is
\begin{equation*}
\mathbf{A}^{-1}
=
\mathbf{I}_K
-\frac{\mathbf{B}}{1+\sum_{k=1}^K({w}_k-{z}_k)}
=
\mathbf{I}_K
-\frac{\mathbf{B}}{\sum_{k=1}^K{w}_k}.
\end{equation*}
By regarding \eqref{eq:ode-hugoniot-vv-pres} as a linear system, solve it as
\beql{eq:solution-ode-vv}
\begin{split}
\left[
\bga{c}
\dfr{dp_1}{d\tau}\\[3mm]
\vdots\\[3mm]
\dfr{dp_K}{d\tau}
\eda
\right]
&=
\mathbf{A}^{-1}
\left[
\begin{array}{c}
-m^2({w}_1-{z}_1) - \dfr{\gamma_1p_1}{\tau} \\[3mm]
\vdots\\[3mm]
-m^2({w}_K-{z}_K) - \dfr{\gamma_Kp_K}{\tau} 
\end{array}
\right]
\\
&=
-\dfr{m^2-\sum_{k=1}^K\gamma_kp_k/\tau}{\sum_{k=1}^K{w}_k}
\left[
\begin{array}{c}
{w}_1-{z}_1\\[3mm]
\vdots\\[3mm]
{w}_K-{z}_K
\end{array}
\right]
-\left[
\begin{array}{c}
\dfr{\gamma_1p_1}{\tau} \\[3mm]
\vdots\\[3mm]
\dfr{\gamma_Kp_K}{\tau} 
\end{array}
\right].
\end{split}
\eeq
By the fact that $
\dfr{ds_k}{d\tau}=0\Longleftrightarrow\dfr{dp_k}{d\tau}=-\dfr{\gm_kp_k}{\tau}
$, \eqref{eq:solution-ode-vv} boils down to \eqref{eq:condition-isentropic-1} or $w_k-z_k=0$, which is just \eqref{eq:condition-isentropic-2}.

For the converse direction, substitute either condition \eqref{eq:condition-isentropic-1} or condition \eqref{eq:condition-isentropic-2} into \eqref{eq:ode-vv}. That leads to the isentropic condition $\frac{dp_k}{d\tau}=-\frac{\gm_kp_k}{\tau}$.
\end{proof}
}

{
\begin{dfntn}
Define
\beql{eq:addmissible-set-1}
\mathcal{A}_1(\tau):=
\left\{(p_1\dots,p_K) \ : \ p_k>0, \ 
\sum_{k=1}^N\dfr{\gm_kp_k}{\tau}>m^2, \text{ for } k=1,\dots,K
\right\},
\eeq
and
\beql{eq:addmissible-set-2}
\mathcal{A}_2(\tau):=
\left\{(p_1\dots,p_K) \ : \ p_k>\mathfrak{q}_k, \text{ for } k=1,\dots,K
\right\},
\eeq
where
\beql{eq:def-degeneration}
\mathfrak{q}_k = p_{k,-}+(\gm_k-1)\rho_-{c_-}^2\tilde{a}_k.
\eeq
The admissible set of the vanishing viscosity Hugoniot curve in the phase space $\{(p_1,\dots,p_K,\tau)\in\mathbb{R}^{K+1}\}$ is defined as
\beql{eq:addmissible-set-total}
\mathcal{A}:=\d\bigcup_{\tau<\tau_-}\Big(\mathcal{A}_1(\tau)
\cap
\mathcal{A}_2(\tau)\Big).
\eeq
\end{dfntn}
\begin{rmrk}
The condition in the definition of $\mathcal{A}_1$ is equivalent to the fact that a shock is subsonic with respect to the upstream fluid.
\end{rmrk}
}

{
\vspace{2mm}
\begin{lmm}
In the admissible set $\mathcal{A}$, the entropy production condition \eqref{eq:entropy-production} holds along the integral curve of \eqref{eq:ode-hugoniot-vv-pres}.
\end{lmm}
\begin{proof}
By \eqref{eq:solution-ode-vv},
\beqs
\dfr{dp_k}{d\tau}
=
-\displaystyle\frac{{w}_k-{z}_k}{\sum_{j=1}^Kw_j}
\left(m^2-\sum_{j=1}^K\dfr{\gamma_jp_j}{\tau}\right)
-\dfr{\gamma_kp_k}{\tau}.
\eeqs
The condition in \eqref{eq:addmissible-set-2} is equivalent to
\beqs
p_k-p_{k,-}>(\gamma_{k}-1)\rho_-{c_-}^2\tilde{a}_{k}, \ \ k=1,\dots,K,
\eeqs
which means, by \eqref{eq:coeff-ode-vv-pres}, ${w}_k-{z}_k<0$.
Furthermore, \eqref{eq:addmissible-set-1} implies
\beqs
m^2-\sum_{k=1}^K\dfr{\gamma_kp_k}{\tau}<0
\eeqs
and, thus, 
\beql{eq:entropy-production-p}
\dfr{dp_k}{d\tau}<-\dfr{\gm_kp_k}{\tau}, \ \ k=1,\dots,K,
\eeq
which, by \eqref{eq:p-to-s-polytropic}, are equivalent to the entropy production condition \eqref{eq:entropy-production}.
\end{proof}
}

{
\vspace{2mm}
\begin{lmm}\label{lem:trajectory}
For all $(p_{1},\dots,p_{K},\tau)\in\mathcal{A}$, the integral curve  of  the ODE system \eqref{eq:ode-hugoniot-vv-pres} initiating from $(p_{1},\dots,p_{K},\tau)$ will remain in the admissible set defined in \eqref{eq:addmissible-set-total}, for $\tau$ close enough to $\tau_-$.
\end{lmm}
The proof is tedious and involves only basic calculations so we put it in Appendix \ref{app:lemma}.
}

\begin{proof}[Sketch of the proof of Theorem~\ref{thm:entropy-vv}]
Now we are ready to prove Theorem \ref{thm:entropy-vv} following the last five lemmata. The philosophy of the proof is that in a sufficiently small interval $\tau<\tau_-$, the Hugoniot curve falls inside the admissible set $\mathcal{A}$. Then, Lemma \ref{lem:trajectory} ensures that the Hugoniot curve will remain entropy productive for a sufficiently weak shock. The proof is henceforth straightforward but tedious and is therefore carried out in the Appendix \ref{app:thm-1}.
\end{proof}

As stated in previous part of this section, for the non-conservative multi-temperature Euler equations, the Hugoniot relation for the total internal energy \eqref{eq:Hugoniot-total} is an energetic constraint rather than a complete shock selection rule. {It determines a Hugoniot hypersurface in the phase space $(p_1,\dots,p_K,\tau)$}, while a shock relation must select an oriented $1$-dimensional curve on this hypersurface. The vanishing viscosity construction supplies such curves through the additional parameter $\frac{\mu_k}{\mu}\in(0,1)$ for $k=1,\dots,K$; as this ratio varies, the corresponding Hugoniot curves sweep out and hence parameterize the Hugoniot hypersurface. In Figure~\ref{fig:surface}, we present an example for the two-temperature model where $k=i,e$ and the pre-shock state is $(\tau_-,p_{i,-},p_{e,-})=(0.25,2,0.1)$ with $\gm_i=1.4$ and $\gamma_e=4.4$. The Hugoniot curve determined by the segment-path Hugoniot relation~\eqref{eq:k-th-Hugoniot-straight} is also shown; it lies on the same hypersurface because the phasic Hugoniot relations sum to \eqref{eq:Hugoniot-total}. This geometric picture emphasizes the source of non-uniqueness: the macroscopic Rankine--Hugoniot relation fixes the surface of energetically admissible states, whereas the missing microscopic or path information chooses a curve on it. An oriented branch on this surface becomes an admissible Hugoniot curve only when it satisfies the criteria listed above, in particular the phasic entropy-production inequalities $\frac{ds_k}{d\tau}<0$.
\begin{figure}[!htb]
\centering
\includegraphics[width=.6\linewidth]{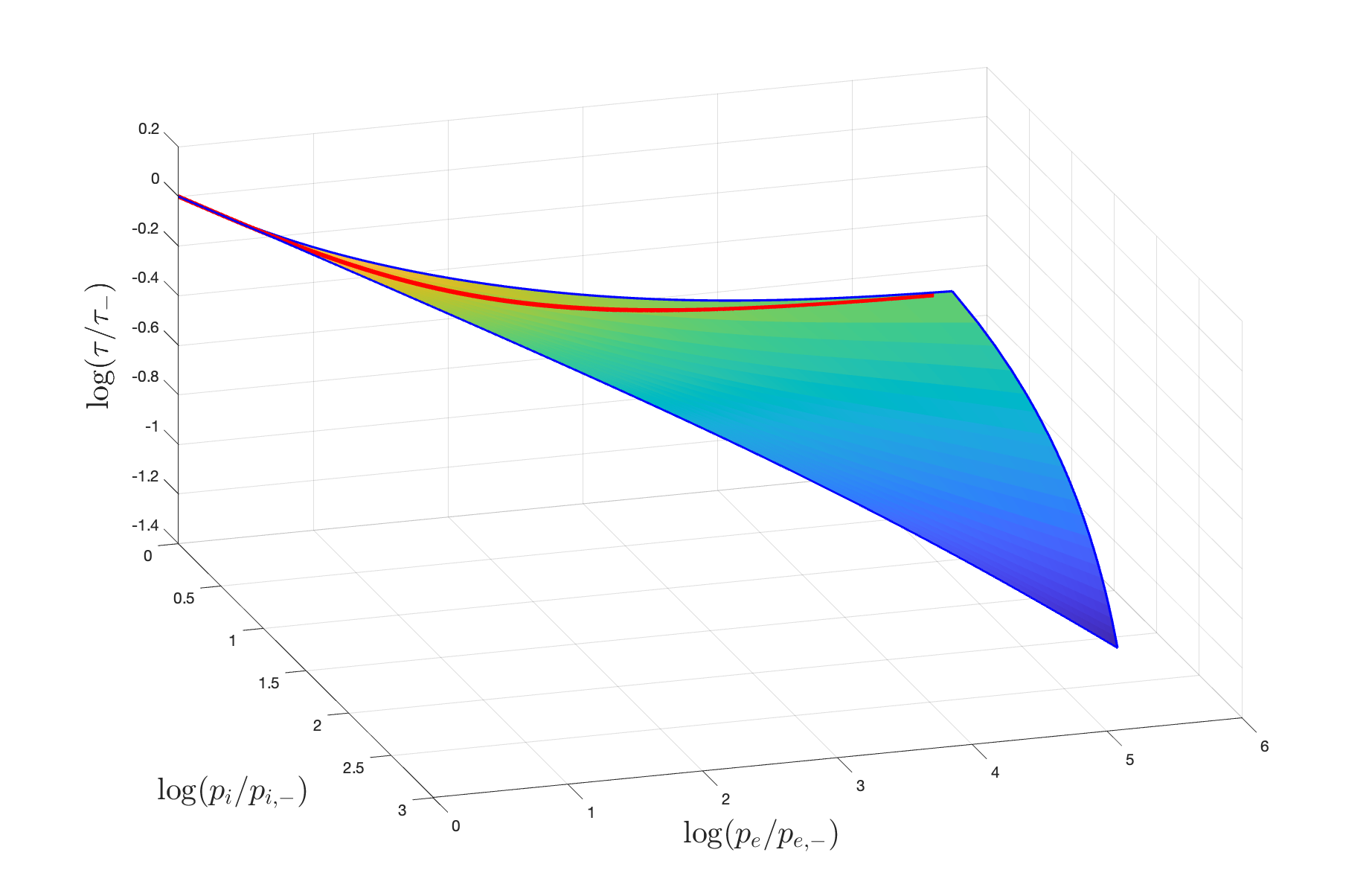}
\caption[small]{The Hugoniot surface parameterized by the vanishing viscosity Hugoniot relation and the segment-path Hugoniot curve (red line).}
\label{fig:surface}
\end{figure}

\vspace{2mm}
\subsection{The fallback to the conventional Hugoniot relation}\label{sec:thermal-equilibrium}
This subsection is dedicated to demonstrate the consistency with classical shock theory under a \emph{formal reduction of the model}. Impose thermal equilibrium as a constitutive constraint,
\begin{equation}\label{eq:temp-equilibrium}
T_1=\cdots=T_K=T,
\end{equation}
and reduce the thermodynamic state space before deriving the shock relations. Summing the phasic Gibbs relations gives
\begin{equation*}
de=T\,ds_{\mathrm{phys}}-p\,d\tau,
\qquad
e=\sum_{k=1}^K e_k,
\quad
s_{\mathrm{phys}}=\sum_{k=1}^K s_k,
\quad
p=\sum_{k=1}^K p_k.
\end{equation*}
A single mixture EOS is defined as $p=p(\rho,e)$; see Appendix~\ref{app:thermal-eq-eos}. The reduced system is therefore the conservative single-energy Euler system. Its shock states satisfy the classical Hugoniot relation
\begin{equation*}
e-e_-+\frac{p+p_-}{2}(\tau-\tau_-)=0.
\end{equation*}

From the geometric viewpoint, the total-energy relation defines a $K$-dimensional hypersurface in $(p_1,\ldots,p_K,\tau)$-space. The $K-1$ equilibrium constraints \eqref{eq:temp-equilibrium} reduce this hypersurface to the one-dimensional Hugoniot curve of the equilibrium mixture EOS. This establishes consistency of the macroscopic framework with the classical conservative baseline.

This formal reduction must be distinguished from restricting a nonequilibrium Hugoniot curve to the equilibrium submanifold. Although both $\mathrm{H}^{\mathrm{seg}}$ and $\mathrm{H}^{\mathrm{vis}}$ satisfy the total-energy relation, their individual phasic closures need not preserve $T_1=\cdots=T_K$ across a nontrivial shock. Establishing convergence of either nonequilibrium shock structure to the equilibrium Hugoniot would require an additional thermal-relaxation analysis and is not claimed here.

\vspace{2mm}
\section{The Riemann problem}\label{sec:riemann-problem}
For the sake of completeness, this section is devoted to the construction of the exact Riemann solver for the multi-temperature model~\eqref{eq:multi-temp-nD} based on the results of~\cite{chalons2005riemann}. We consider the initial value problem with the initial state
\beqs
\bU(x,0)=\left\{
\bga{ll}
\bU_L, & x<0,\\
\bU_R, & x>0.
\eda
\right.
\eeqs

Since all waves are either genuinely nonlinear or linearly degenerate we successively handle rarefaction waves, shocks and contact discontinuities. 
Rarefaction waves are resolved in the standard way by analyzing the Riemann invariants of the genuinely nonlinear waves, i.e.~waves corresponding to the eigenvalues $\lambda_{\pm}$ defined in~\eqref{eq:multi-temp-nD-eigenval}. The standard technique to compute the shock curves, where  a unique Hugoniot curve corresponds to the Rankine--Hugoniot conditions, is not applicable to the  non-conservative multi-temperature Euler equations. Therefore, a single Hugoniot relation must be selected to obtain a unique shock wave as discussed  in Section~\ref{sec:wave-right}. Therein, we have considered to special cases: a segment path (Section~\ref{sec:segment-path-hugoniot}) and the  vanishing viscosity ansatz (Section~\ref{sec:vanishing-visc-hugoniot}).  Both approaches yield physically meaningful shock curves that  will be used for the Riemann problem solution in what follows.

\vspace{2mm}
\paragraph{\textbf{Rarefaction waves}}
A short calculation shows that the Riemann invariants for the 1-wave, the left acoustic wave (corresponding to $\lambda_-$), are $s_k$ {with} $k=1,\ldots, K$, and
\begin{equation*}\label{eq:chara-1st}
u_n+\int^\rho\dfr{c(\om;s_1, \ldots,s_K)}{\om}d\om.
\end{equation*}
Similarly, the ones for the $(K+2)$-wave, the right acoustic wave (corresponding to $\lambda_+$), are $s_k$ {with} $k=1,\ldots, K$, and
\begin{equation*}\label{eq:chara-3rd}
u_n-\int^\rho\dfr{c(\om;s_1,\ldots,s_K)}{\om}d\om.
\end{equation*}
Since the Riemann invariants are constant across the rarefaction wave, we have
\begin{equation}\label{eq:rarefaction-post}
(u_n)^{\text{rw}, \pm}\coloneq u_{n,-}\mp\int^{\rho}_{\rho_-}\dfr{c(\om,s_1, \ldots,s_K)}{\om}d\om, \ \ s_k \ev s_{k,-}, \ \ k=1,\ldots,K,
\end{equation}
where the minus sign is taken for left rarefaction waves and $\bV_-$ is the pre-wave state.
Next, we rewrite the first equation in \eqref{eq:rarefaction-post} by substituting the integration variable to the specific volume $\tau$ 
\begin{equation}\label{eq:u-tau-left} 
\begin{split}
(u_n)^{\text{rw}, \pm}=u_{n,-}\pm\d\int^{\tau}_{\tau_-}\dfr{c(\eta;s_1,\ldots,s_K)}{\eta}d\eta=u_{n,-}\pm\d\int^{\tau}_{\tau_-}\dfr{c(\eta;s_{1,-}, \ldots,s_{K,-})}{\eta}d\eta.
\end{split}
\end{equation}
where the positive sign is taken for left rarefaction waves and vice versa.

The phasic pressures across the rarefaction wave can be expressed as a function of the post-wave specific volume and the pre-wave entropy, i.e.
\begin{equation}\label{eq:pi-tau-pe-tau-left}
(p_k)^{\text{rw},\pm} = \mathcal{P}_k(\tau; \tau_-, s_{k,-}), \quad  k = 1.\ldots, K,
\end{equation}
and are, therefore, solely determined by the last $K$ equations in \eqref{eq:rarefaction-post}. 
In particular, for polytropic gases equations \eqref{eq:u-tau-left} and \eqref{eq:pi-tau-pe-tau-left}  read
\begin{equation*}\label{eq:u-tau-left-gamma}
\bga{l}
\d (u_n)^{\text{rw},\pm}
=u_{n,-}\mp\int^{\tau}_{\tau_-}\left[ \sum_{k=1}^K\gm_k\varsigma_{k,-}\eta^{-(\gm_k+1)} \right]^{\frac 12}d\eta, \quad (p_k)^{\text{rw}, \pm}_k=\tau^{-\gm_k}p_{k,-}{\tau_-}^{\gm_k}, \quad k = 1,\ldots, K.
\eda
\end{equation*}

\vspace{2mm}
{\paragraph{\textbf{Shock waves}}
Shock waves are entirely characterized by the chosen Hugoniot relations, which serve as the necessary closure for the non-conservative system. As established in Section~\ref{sec:wave-right}, the selection of this closure fundamentally dictates the physical nature of the shock transition. To construct the Riemann solver, we evaluate two distinct approaches: the segment-path Hugoniot \eqref{eq:k-th-Hugoniot-straight} and the vanishing viscosity Hugoniot \eqref{eq:Hugoniot-vv}.} 

Using the polytropic gas EOS~\eqref{eq:eos}, the pressure along the 1-wave and the $(K+2)$-wave for the geometric segment path is {directly computed from~\eqref{eq:k-th-Hugoniot-straight}, derived in Section~\ref{sec:segment-path-hugoniot}, as follows:}
\begin{equation}\label{eq:p-tau-seg}
(p_k)^{\text{seg}}(\tau; \bV_-)=\dfr{(\gm_k+1)\tau_--(\gm_k-1)\tau}{(\gm_k+1)\tau-(\gm_k-1)\tau_-}p_{k,-}.
\end{equation}
The pressure for the vanishing viscosity, $(p_k)^{\text{vis}}$, on the contrary cannot be cast into an explicit expression and has to be obtained implicitly by finding stagnation points of the {ODE System~\eqref{eq:ode-vv} derived in Section~\ref{sec:vanishing-visc-hugoniot}.}
Having determined the pressure across the shock wave, the Rankine-Hugoniot relation~\eqref{eq:p-u-shock-0} provides the according velocity
\begin{equation}\label{eq:u_n-right}
u_n = u_{n,-} \mp \sqrt{\left[\sum_{k=1}^K(p_k - p_{k,-})\right](\tau - \tau_-) }.\\
\end{equation}

\vspace{2mm}
\paragraph{\textbf{Contact discontinuities}}
The eigenvalue $u_n$ has multiplicity $d+K-1$. In one space dimension, the associated contact manifold is characterized by
\begin{equation}\label{eq:Contact-curves}
\mathcal C(\bV_-)
=\left\{\bV:\ u_n=u_{n,-},\quad
\sum_{k=1}^Kp_k=\sum_{k=1}^Kp_{k,-}\right\}.
\end{equation}
Thus normal velocity and total pressure are continuous, while density and the distribution of total pressure among the species may jump. In particular, a contact discontinuity is not required to preserve each phasic entropy. Tangential velocity jumps supply the additional linearly degenerate modes in several space dimensions.


\vspace{2mm}
\paragraph{\textbf{Lax curves and the solution of the Riemann problem}}
{Now we have all necessary ingredients to solve the Riemann problem. For details on existence of solutions, in particular for the vanishing viscosity approach we refer to to~\cite{chalons2005riemann}.
Let $\zeta_1=1$ and $\zeta_{K+2}=-1$. The rarefaction curves are
\begin{equation}\label{eq:Rarefaction-curves}
    R_i(\tau; \bV_-) \coloneq \begin{cases}
        u_n = u_{n,-} + \zeta_i\displaystyle\int^{\tau}_{\tau_-}\left[ \sum_{k=1}^K\gm_k\varsigma_{k,-}\eta^{-(\gm_k+1)} \right]^{\frac 12}d\eta,\\
        p_k = \tau^{-\gamma_k} p_{k,-} \tau_-^{\gamma_k}
    \end{cases}
    \qquad i\in\{1,K+2\}.
\end{equation}

Let $\chi_1=-1$ and $\chi_{K+2}=1$. The shock curves are
\begin{equation}\label{eq:Shock-curves}
    S_i(\tau; \bV_-) \coloneq \begin{cases}
        u_n = u_{n,-} + \chi_i\sqrt{(p - p_{-})(\tau_- - \tau) },\\
        p_k = (p_k)^{\text{path}}(\tau, \bV_-)  
    \end{cases}
    \qquad i\in\{1,K+2\},
\end{equation}
where $p = \sum_{k=1}^Kp_k$ and $(p_k)^{\text{path}}$ is determined by the selected Hugoniot relation, for example $(p_k)^{\text{seg}}$ or $(p_k)^{\text{vis}}$.}

The solution of the Riemann problem is now acquired by the standard procedure for the Euler equations~\cite{toro2013riemann} by solving
\begin{equation*}
\begin{split}
&(u_n)^+_1(\tau^+_{1,*}, \bV_L) = (u_n)^-_{K+2}(\tau^-_{K+2,*}, \bV_R),\\
&\sum_{k=1}^K(p_k)^+_1(\tau^+_{1,*}, \bV_L) = \sum_{k=1}^K(p_k)^-_{K+2}(\tau^-_{K+2,*}, \bV_R),
\end{split}
\end{equation*}
which yields $\tau^+_{1,*}$ and $\tau_{K+2,*}$.

\vspace{2mm}
\section{Numerical Schemes}\label{sec:scheme}
In Section~\ref{sec:wave-right}, two Hugoniot relations of the non-conservative governing PDEs \eqref{eq:multi-temp-nD} in one space dimension have been analyzed.
Numerical schemes have been proposed in~\cite{strt-prsv,scheme-e-vv}. The discussion below should be regarded as a heuristic and formal comparison, not as a proof of the limiting behavior of these discretizations under mesh refinement. We do not establish that the numerical approximations approach a particular shock solution; rather, we examine which Hugoniot relation is encoded by the algebraic form of each scheme when applied to simplified Riemann data containing prescribed shock waves.

Instead of the Eulerian framework, we consider the Lagrangian version of structure-preserving and vanishing viscosity schemes, since this formulation makes the relation between cell updates, numerical-viscosity terms, and shock paths more transparent.
In order to compare the results, we confine ourselves to two phases and one spatial dimension. The governing PDEs for the two-temperature compressible flow in the one-dimensional Lagrangian framework are
\beql{eq:pde-lag}
\begin{split}
&\rho\dfr{d}{dt}\dfr{1}{\rho}-\dfr{\pt u}{\pt x}=0,\\
&\rho\dfr{du}{dt}+\dfr{\pt p_i}{\pt x}+\dfr{\pt p_e}{\pt x}=0,\\
&\rho\dfr{de_k}{dt}+p_k\dfr{\pt u}{\pt x}=0, \ \ k=1,\dots,K.
\end{split}
\eeq

\vspace{2mm}
\subsection{The structure-preserving scheme}
Following the ansatz in \cite{strt-prsv}, the structure-preserving scheme for \eqref{eq:pde-lag} is
\begin{equation*}\label{eq:sp-scheme-0}
\begin{split}
&\overline{u}_j^{n+1} = \bar{u}_j^n-\dfr{\Dt}{m_j}(p^*_{j+\frac 12}-p^*_{j-\frac 12}),\\
&{\overline{(e_k)}_j^{n+1}=\overline{(e_k)}_j^n-\dfr{\Dt}{m_j} \Big\{(p_k)^*_{j+\frac 12}u^*_{j+\frac 12}-(p_k)^*_{j-\frac 12}u^*_{j-\frac 12}-\dfr{\overline{u}_j^{n+1}+\overline{u}_j^n}{2}\Big[(p_k)^*_{j+\frac 12}-(p_k)^*_{j-\frac 12}\Big] \Big\},\,\,k=1,\dots,K,}
\end{split}
\end{equation*}
where $q^*_{j\pm\frac12}$ are Riemann solutions at cell interfaces, $m_j$ is the mass of the cell $(x_{j-\frac 12},x_{j+\frac 12})$ which is constant in time, and
\beqs
\overline{q}_j^{n+\nu}=\dfr{1}{x_{j+\frac 12}^{n+\nu}-x_{j-\frac 12}^{n+\nu}}
\displaystyle\int_{x_{j-\frac 12}^{n+\nu}}^{x_{j+\frac 12}^{n+\nu}}q(x,t^n+\nu\Dt)dx, \ \ \nu=0,1.
\eeqs

First, we consider a simplified case of
a $(K+2)$-shock wave emanates from $x_L(0)$ moving with the uniform velocity $\sigma$ at $t=0$.
Consider a piece-wise constant initial data
\begin{equation}\label{eq:initial-data}
\bU(x,0) = \left\{
\bga{ll}
\bU_L, & x<0,\\
\bU_R, & x>0,
\eda
\right.
\end{equation}
where $\bU_L$ and $\bU_R$ are connected by a $(K+2)$-shock under the assumption of an arbitrary path with shock speed $\sigma$.
Consider a Lagrangian cell $I(t)=(x_L(t),x_R(t))$ where $x_L(0)=0$. 
Assume that at $\Dt$, the shock hits the right cell boundary $x_R$, i.e.
\begin{equation}\label{eq:hit-right}
x_L(0)+\sigma\Dt=x_R(\Dt).
\end{equation}
Thus the Lagrangian evolution in the cell reads
\begin{equation}\label{eq:sp-scheme}
\bga{l}
\overline{u}^{n+1} = u_R-\dfr{\Dt}{m_I}(p_R-p_L),\\[2.5mm]
{\overline{(e_k)}^{n+1}=e_{k,R}-\dfr{\Dt}{m_I} \Big[p_{k,R}u_R-p_{k,L}u_L-\dfr{u_L+u_R}{2}(p_{k,R}-p_{k,L}) \Big],\,\,k=1,\dots,K,}
\eda
\end{equation}
where $m_I$ is the mass of the fluid in the cell.
The second equation in \eqref{eq:sp-scheme} directly yields
\begin{equation*}
{\overline{(e_k)}^{n+1}=e_{k,R}-\dfr{\Dt}{m_I}\dfr{p_{k,R}+p_{k,L}}{2}(u_R-u_L),\,\,k=1,\dots,K,}
\end{equation*}
which can be written as
\begin{equation}\label{eq:pre-Hugoniot-i-sp}
{\overline{(e_k)}^{n+1}=e_{k,R}-\dfr{\Dt}{m_I}\dfr{p_{k,R}+p_{k,L}}{2}\big[(u_R-\sigma)-(u_L-\sigma)\big],\,\,k=1,\dots,K.}
\end{equation}
For a {(K+2)-shock}, we have $\rho_R(\sigma-u_R)\Dt=\rho_L(\sigma-u_L)\Dt=m_I$, that transforms \eqref{eq:pre-Hugoniot-i-sp} into
\begin{equation}\label{eq:Hugoniot-i-sp}
{\overline{(e_k)}^{n+1}=e_{k,R}+\dfr{p_{k,R}+p_{k,L}}{2}(\tau_R-\tau_L),\,\,k=1,\dots,K.}
\end{equation}
By the Hugoniot relation \eqref{eq:k-th-Hugoniot-straight}, {$\overline{(e_k)}^{n+1} = e_{k,L}$ for $k=1,\dots,K$}.

\begin{rmrk}
    Hence, the derivation from equation \eqref{eq:sp-scheme} to \eqref{eq:Hugoniot-i-sp} indicates that the structure-preserving scheme proposed in~\cite{strt-prsv} formally recovers the shock path connecting the pre- and post-shock states by a straight line.
\end{rmrk}

\vspace{2mm}

In general, the single shock emanating from $x_L(0)$ is not likely to hit $x_R(\Dt)$ due to the CFL condition. In this case, we firstly derive the Hugoniot relation for the total internal energy $e = e_i + e_e$.
The energy-conservative evolution for the total internal energy reads
\begin{equation}\label{eq:sp-scheme-e}
\overline{e}^{n+1}=e_R-\dfr{\Dt}{m_I}\Big[p_Ru_R-p_Lu_L-\dfr{\overline{u}^{n+1}+\overline{u}^{n}}{2}(p_R-p_L)\Big].
\end{equation}
By the conservation of the momentum it holds
\begin{equation*}
\overline{u}^n = u_R, \ \ \ \overline{u}^{n+1}=\dfr{m_l}{m_I}u_L + \dfr{m_r}{m_I}u_R,
\end{equation*}
where
\begin{equation}\label{eq:fractional-mass}
m_l=\rho_R(\sigma-u_R)\Dt=\rho_L(\sigma-u_L)\Dt, \ \ \ m_r=m_I-m_l.
\end{equation}
Substituting \eqref{eq:fractional-mass} into \eqref{eq:sp-scheme-e}, we have
\begin{equation}\label{eq:sp-scheme-e-2}
\bga{l}
\overline{e}^{n+1}=\dfr{m_l}{m_I}\Big\{e_R-\dfr{\Dt}{m_l}\big[p_Ru_R-p_Lu_L-\dfr{u_L+u_R}{2}(p_R-p_L)\big]\Big\} + \dfr{m_r}{m_I}e_R\\[3mm]
\ \ \ \ \ \ \ \ \ \ \ \ \ \ \ \ \ \ \ \ \ \ \ \ \ \ \ \ \ \ \ \ \ \ \ \ +\dfr{m_l}{m_I}\dfr{\Dt}{m_l}\Big[(\dfr{m_l}{2m_I}-\dfr{1}{2})u_L+\dfr{m_r}{2m_I}u_R\Big](p_R-p_L).
\eda
\end{equation}
Using the total Hugoniot relation \eqref{eq:Hugoniot-total}, we simplify the respective terms in braces to $e_L$ and $e_R$, i.e.
\begin{equation}\label{eq:sp-scheme-e-3}
\overline{e}^{n+1}=\dfr{m_l}{m_I}e_L + \dfr{m_r}{m_I}e_R  +\dfr{\Dt}{m_I}\dfr{m_r}{2m_I}(u_R-u_L)(p_R-p_L).
\end{equation}
The Rakine-Hugoniot relation of the momentum equation reads
\begin{equation*}
\rho_R(u_R-\sigma)^2+p_R=\rho_L(u_L-\sigma)^2+p_L,
\end{equation*}
that implies
\begin{equation*}
\Dt(p_R-p_L)=m_l(u_R-u_L).
\end{equation*}
Substituting the above equation into \eqref{eq:sp-scheme-e-3} results in
\begin{equation}\label{eq:Hugoniot-e-tot-general}
\overline{e}^{n+1}=\dfr{m_l}{m_I}e_L + \dfr{m_r}{m_I}e_R  +\dfr{m_lm_r}{2{m_I}^2}(u_R-u_L)^2.
\end{equation}
The equation \eqref{eq:Hugoniot-e-tot-general} indicates that $e^{n+1}$ amounts to the weighted average of $e_L$ and $e_R$ and an additional numerical viscosity term, where $e_L$ is the post-shock state connected with $e_R$.
For the fractional internal energy $e_k$, a similar derivation leads to
\begin{equation}\label{eq:Hugoniot-e-frac-sp}{
\overline{(e_k)}^{n+1}=\dfr{m_l}{m_I}e_{k,L} + \dfr{m_r}{m_I}e_{k,R}  +\dfr{m_lm_r}{2{m_I}^2}(u_R-u_L)^2\dfr{p_{k,R}-p_{k,L}}{p_R-p_L},\,\,k=1,\dots,K.}
\end{equation}

\vspace{2mm}
\subsection{The vanishing viscosity scheme}
The vanishing viscosity scheme proposed in \cite{scheme-e-vv} is a prediction-correction method. The prediction step is
\begin{equation*}\label{eq:vv-scheme-est}
\begin{split}
&\overline{u}_j^{n+1} = \bar{u}_j^n-\dfr{\Dt}{m_j}(p^*_{j+\frac 12}-p^*_{j-\frac 12}),\\
&{\overline{(e_k)}_j^{n+1,-}=(e_k)_j^{n}-\dfr{1}{m_j}\left[\d\int_{t^n}^{t^{n+1}}\int_{I_j}p_k\dfr{\pt u}{\pt x}dxdt+\sum_{\text{shocks}\in I_j}\rho_-(u_--\sigma)(e_{k,+}-e_{k,-})\right],\,\,k=1,\dots,K,}
\end{split}
\end{equation*}
where $\sigma$ is the shock speed and $q_\pm$ is the pre- and post-shock states.
The total energy is evolved in a conservative manner,
\begin{equation*}\label{eq:vv-scheme-E}
\bar{E}_j^{n+1} = \bar{E}_j^n-\dfr{\Dt}{m_j}(u^*_{j+\frac 12}p^*_{j+\frac 12}-u^*_{j-\frac 12}p^*_{j-\frac 12}),
\end{equation*}
In the correction step the internal energies are updated as
\begin{equation*}\label{eq:vv-scheme-corr}
{\overline{(e_k)}_j^{n+1}=\overline{(e_k)}_j^{n+1,-}+\dfr{\mu_k}{\mu}\Delta E,\,\,k=1,\dots,K,}
\end{equation*}
where {$\mu=\sum_{k=1}^K\mu_k$ and $\Delta E=\bar{E}_j^{n+1}-\frac{1}{2}(\rho u^2)_j^{n+1}-\sum_{k=1}^K\overline{(e_k)}_j^{n+1,-}$}.
In this case, the right-moving $(K+2)$-shock connecting the two states follows the assumption of the vanishing viscosity shock path with the shock speed denoted by $\sigma$.
Once again, we firstly consider the simplified case of a single shock wave \eqref{eq:hit-right}.

For the simplified case, following \cite[eq. (4.6)]{scheme-e-vv}, the vanishing viscosity evolution for internal energies is
\begin{equation}\label{eq:vv-scheme-punch}
{(e_k)^{n+1,-}=e_{k,R}+\dfr{\Dt}{m_I}\rho_R(u_R-\sigma)(e_{k,R}-e_{k,L}),\,\,k=1,\dots,K.}
\end{equation}
Since $\rho_-(u_R-\sigma)\Dt=-m_I$, \eqref{eq:vv-scheme-punch} implies that $(e_k)^{n+1,-}=e_{k,L}$ for $k=1,\dots,K$.

Next, we consider the general case.
For the total internal energy equation~\eqref{eq:Hugoniot-e-tot-general} holds. For the fractional internal energies we have
\begin{equation*}
{(e_k)^{n+1,-}=e_{k,R}+\dfr{\Dt}{m_I}\rho_R(u_R-\sigma)(e_{k,R}-e_{k,L}),\,\,k=1,\dots,K.}
\end{equation*}
In particular, since $\rho_-(u_R-\sigma)\Dt=-m_l$, we obtain
\begin{equation*}
{(e_k)^{n+1,-}=\dfr{m_l}{m_I}e_{k,L}+\dfr{m_r}{m_I}e_{k,R},\,\,k=1,\dots,K.}
\end{equation*}
As expected, the predicted fractional internal energies at $t^{n+1}$ are weighted averages of pre- and post-shock values. However, due to \eqref{eq:Hugoniot-e-tot-general},
\begin{equation*}
{e^{n+1} - \sum_{k=1}^K(e_k)^{n+1,-}= \dfr{m_lm_r}{2{m_I}^2}(u_R-u_L)^2 > 0,}
\end{equation*}
and the correction step is
\begin{equation}\label{eq:Hugoniot-e-frac-vv}
{(e_k)^{n+1} = (e_k)^{n+1,-} + \dfr{\mu_k}{\mu}\dfr{m_lm_r}{2{m_I}^2}(u_R-u_L)^2,\,\,k=1,\dots,K.}
\end{equation}
On one hand, the above correction step can be regraded as a distribution of the numerical viscosity to fractional internal energies. On the other hand, the prediction and correction steps can be regarded as an operator splitting method for the Navier-Stokes equations.

\vspace{2mm}
\begin{rmrk}
The difference between the structure-preserving scheme and the vanishing viscosity scheme is how they distribute the numerical viscosity. From \eqref{eq:Hugoniot-e-frac-sp}, the structure-preserving scheme distributes the numerical viscosity according to the normal fractional pressures. From \eqref{eq:Hugoniot-e-frac-vv}, the vanishing viscosity scheme distributes the  numerical viscosity  according to the physical viscosity coefficients.
\end{rmrk}

\vspace{2mm}
\subsection{Numerical experiments}\label{sec:numer}
In this subsection we present several numerical experiments that are used to verify the performance of the structure-preserving scheme and the vanishing viscosity scheme.

{
\vspace{2mm}
\subsubsection{Accuracy test with three-temperature Euler equations}\label{sec:3t}
We consider the three-temperature Euler equations where $k=i,e,r$.
The subscripts $i$, $e$, and $r$ stands for ions, electrons, and the radiation, respectively.
The fluid is initially at rest with the initial density
\begin{equation*}
\rho_0(x)=1+\frac{1}{5}\sin(2\pi x).
\end{equation*}
The initial pressures are determined by isentropic condition,
\begin{equation*}
\varsigma_{i,0}(x)=\frac{p_{i,0}(x)}{{\rho_0(x)}^{\gamma_i}}\equiv3,
 \,\, \varsigma_{e,0}(x)=\frac{p_{e,0}(x)}{{\rho_0(x)}^{\gamma_i}}\equiv2,
 \,\, \varsigma_{r,0}(x)=\frac{p_{r,0}(x)}{{\rho_0(x)}^{\gamma_i}}\equiv1.
\end{equation*}
As long as the flow remains smooth, the three entropies are constant and can be used to assess the numerical accuracy of schemes listed in the previous two subsections. We run simulations with the two numerical scheme to the time $T=0.2$. The numerical errors of entropies and corresponding convergence orders are listed in Table \ref{tab:accuracy-test-1d-fractional}. Expected numerical orders are verified.

\begin{table*}[!htbp]
  \centering
  \caption[small]{$L_1$ and $L_\infty$ errors of $\varsigma_i$, $\varsigma_e$, and $\varsigma_r$, and corresponding convergence rates.}
\footnotesize
  \begin{tabular}{|c|cccc|cccc|}
    \hline
                     &    \multicolumn{8}{|c|}{errors of $\varsigma_i$}       \\\cline{2-9}
 mesh size    &    \multicolumn{4}{|c|}{structure-preserving scheme}     &  \multicolumn{4}{|c|}{vanishing viscosity scheme}   \\\cline{2-9}
                    & $L_1$ error &  order   & $L_\infty$ error & order  & $L_1$ error &  order   & $L_\infty$ error & order  \\\hline
1/20 & 9.42e-03 & - & 1.04e-02 & - &1.21e-02 & - & 1.37e-02 & -  \\
1/40 & 5.64e-03 & 0.74 & 6.85e-03 & 0.60 &7.04e-03 & 0.78 & 8.54e-03 & 0.68  \\
1/80 & 3.13e-03 & 0.85 & 4.06e-03 & 0.75 &3.83e-03 & 0.88 & 4.94e-03 & 0.79  \\
1/160 & 1.65e-03 & 0.92 & 2.24e-03 & 0.86 &2.01e-03 & 0.93 & 2.70e-03 & 0.87  \\
1/320 & 8.50e-04 & 0.96 & 1.18e-03 & 0.92 &1.03e-03 & 0.97 & 1.42e-03 & 0.93  \\\hline
                     &    \multicolumn{8}{|c|}{errors of $\varsigma_e$}       \\\cline{2-9}
 mesh size    &    \multicolumn{4}{|c|}{structure-preserving scheme}     &  \multicolumn{4}{|c|}{vanishing viscosity scheme}   \\\cline{2-9}
                    & $L_1$ error &  order   & $L_\infty$ error & order  & $L_1$ error &  order   & $L_\infty$ error & order  \\\hline
1/20 & 1.25e-02 & - & 1.38e-02 & - &7.39e-03 & - & 1.06e-02 & -  \\
1/40 & 7.49e-03 & 0.74 & 9.08e-03 & 0.60 &4.31e-03 & 0.78 & 6.55e-03 & 0.70  \\
1/80 & 4.14e-03 & 0.85 & 5.37e-03 & 0.76 &2.35e-03 & 0.88 & 3.69e-03 & 0.83  \\
1/160 & 2.19e-03 & 0.92 & 2.96e-03 & 0.86 &1.23e-03 & 0.93 & 1.97e-03 & 0.90  \\
1/320 & 1.13e-03 & 0.96 & 1.56e-03 & 0.92 &6.30e-04 & 0.97 & 1.02e-03 & 0.95  \\\hline 
                     &    \multicolumn{8}{|c|}{errors of $\varsigma_r$}       \\\cline{2-9}
 mesh size    &    \multicolumn{4}{|c|}{structure-preserving scheme}     &  \multicolumn{4}{|c|}{vanishing viscosity scheme}   \\\cline{2-9}
                    & $L_1$ error &  order   & $L_\infty$ error & order  & $L_1$ error &  order   & $L_\infty$ error & order  \\\hline
1/20 & 2.49e-03 & - & 2.73e-03 & - &1.28e-03 & - & 1.37e-02 & -  \\
1/40 & 1.49e-03 & 0.74 & 1.81e-03 & 0.60 &7.49e-04 & 0.77 & 1.25e-03 & 0.69  \\
1/80 & 8.27e-04 & 0.85 & 1.07e-03 & 0.75 &4.09e-04 & 0.87 & 7.01e-04 & 0.84  \\
1/160 & 4.37e-04 & 0.92 & 5.92e-04 & 0.86 &2.15e-04 & 0.93 & 3.72e-04 & 0.91  \\
1/320 & 2.25e-04 & 0.96 & 3.13e-04 & 0.92 &1.10e-04 & 0.96 & 1.92e-04 & 0.95  \\\hline 
  \end{tabular}
  \label{tab:accuracy-test-1d-fractional}
\end{table*}

\vspace{2mm}
\subsubsection{Three-temperature shock tube problem}\label{sec:3t-shock}
We consider the three-temperature version of the Sod shock tube problem.
The fluid is initially at rest with the thermodynamic state
\begin{equation*}
(\rho_0(x),p_{i,0}(x),p_{e,0}(x),p_{r,0}(x))=\left\{
\begin{array}{ll}
1,\,\,0.5,\,\,0.3.\,\,0.2, &x\in(0,0.5),\\
0.125,\,\,0.02,\,\,0.05,\,\,0.03, &x\in(0.5,1).\\
\end{array}
\right.
\end{equation*}
The thermodynamics parameters are $\gamma_i=1.4$, $\gamma_e=\frac{5}{3}$, $\gamma_r=\frac{4}{3}$, $\mu_i=200$, $\mu_e=1$, and $\mu_r=0.1$. We run simulation with both structure-preserving and vanishing viscosity schemes to the time $T=0.2$. Comparisons between numerical solutions and exact solutions on ion, electron, and radiation pressures are displayed in Figure \ref{fig:3t}. Numerical solutions by both schemes match corresponding exact solutions.
\begin{figure}[!htb]
\centering
\includegraphics[width=.305\linewidth]{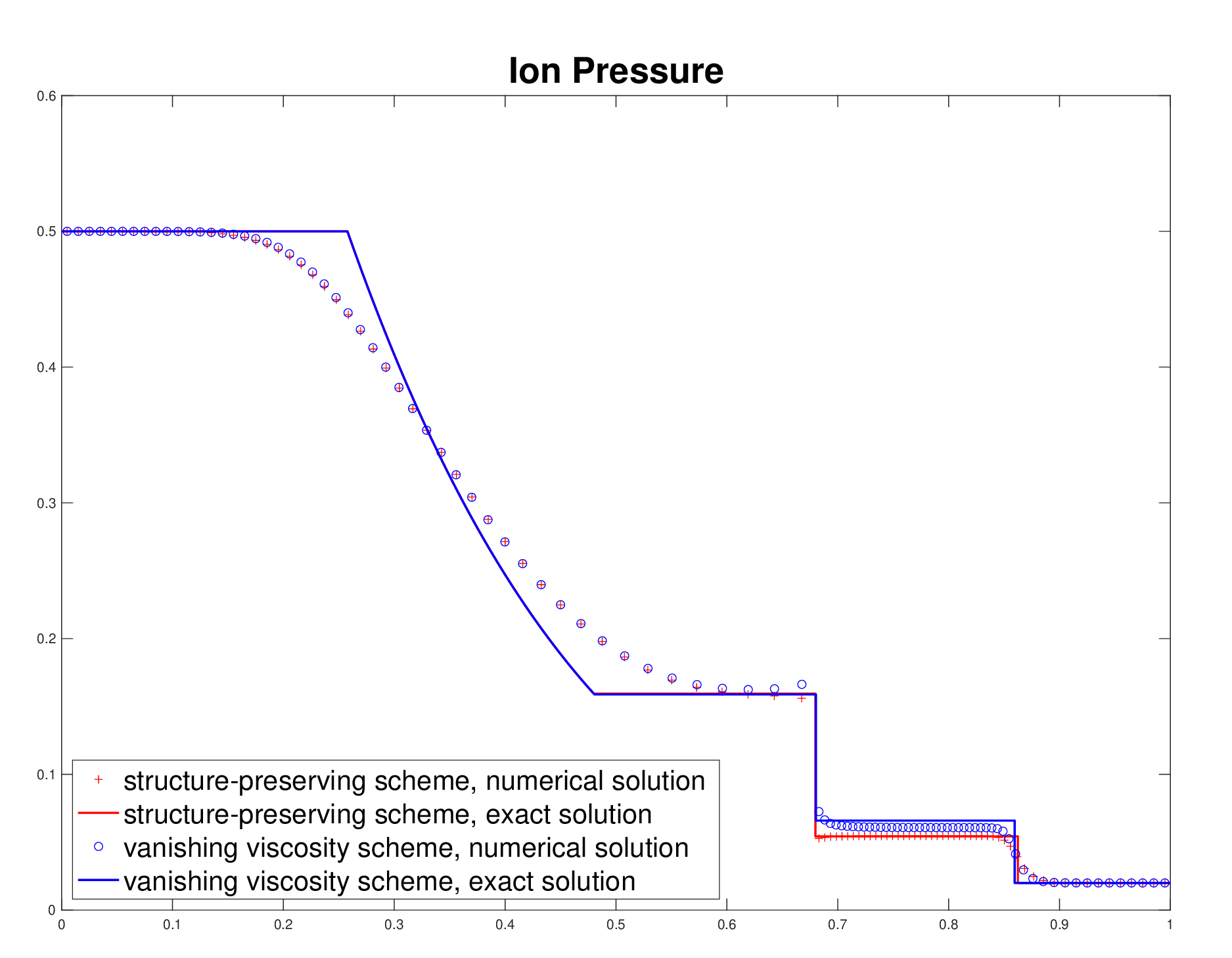}
\includegraphics[width=.305\linewidth]{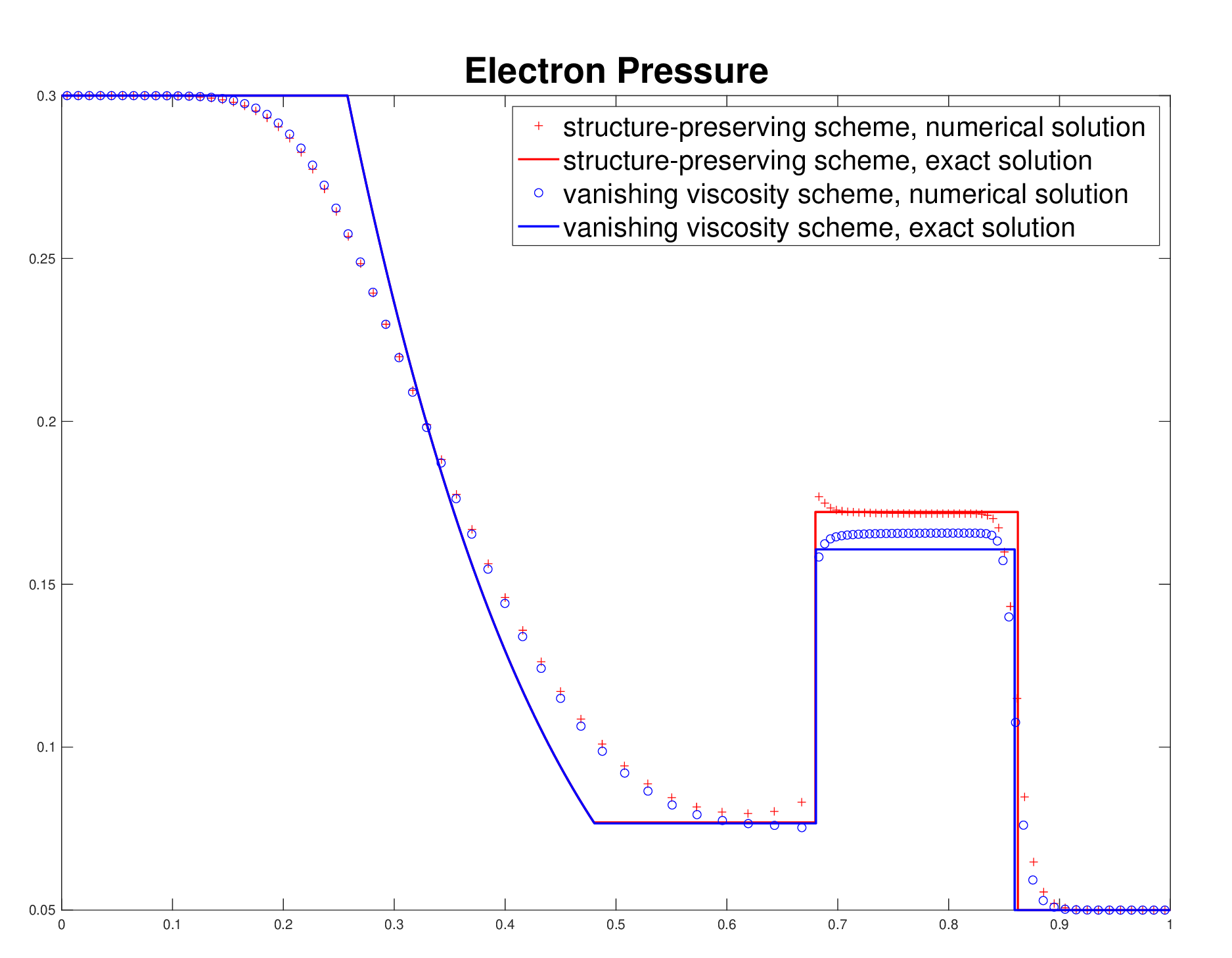}
\includegraphics[width=.305\linewidth]{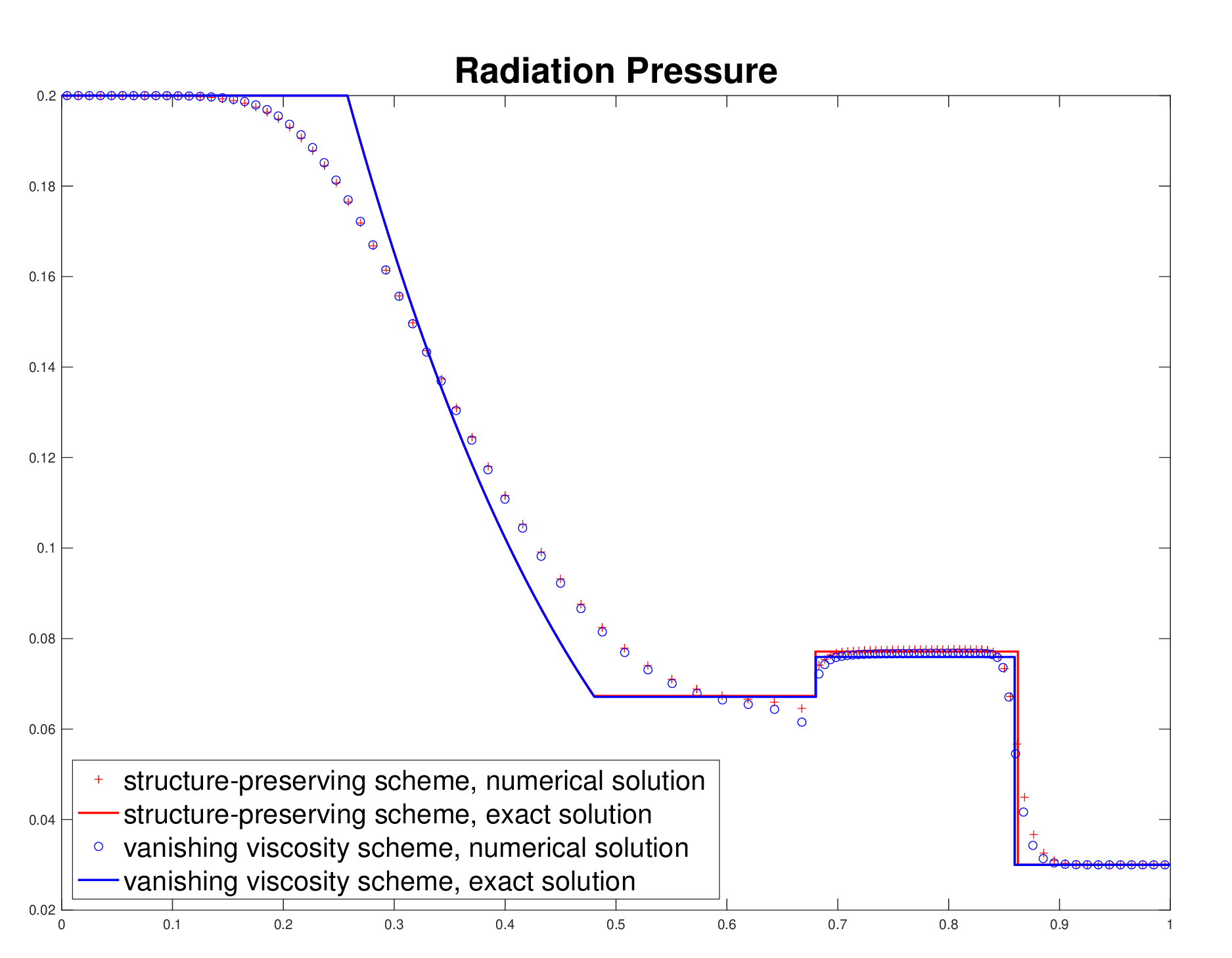}
\caption[small]{Three-temperature shock tube problem. The ion pressure (left), the electron pressure (middle), and the radiation pressure (right).}
\label{fig:3t}
\end{figure}
}

\vspace{2mm}
\subsubsection{Two-temperature double shock Riemann problems}\label{sec:double-shock}
With $k=i,e$, consider three two-temperature double shock Riemann problems~\cite{chalons2005riemann}. The initial data is listed in Table \ref{tab:acc-2-shock}.

\begin{table}[!htbp]
  \centering
  {\small
    \begin{tabular}{| c | c | c | c | c | c | c | c | c | c |}
      \hline
      Test case  &  final time  &$\gm_i$  &  $\gm_e$  &  $\mu_i$  &  $\mu_e$  &  $\rho$  &  $u$  &  $p_i$  &  $p_e$\\\hline
      \multirow{2}{*}{A} & \multirow{2}{*}{$0.26$} & \multirow{2}{*}{$1.4$} & \multirow{2}{*}{$\dfr{5}{3}$} & \multirow{2}{*}{100}  &  \multirow{2}{*}{ 1 }  & $1$ & $1$ &  $1$ & $0.6$ \\
        \cline{7-10}
                 &                        &                        &                               &                       &                      & $1.05518$ & $-0.88895$ & $0.15031$ & $0.33869$\\
      \hline
      \multirow{2}{*}{B} & \multirow{2}{*}{$0.2$} & \multirow{2}{*}{$1.4$} & \multirow{2}{*}{$\dfr{5}{3}$} & \multirow{2}{*}{ 1 }  &  \multirow{2}{*}{ 1 }  & $1$ & $1$ &  $1$ & $0.6$ \\
        \cline{7-10}
                 &                        &                        &                               &                       &                      & $1.92678$ & $-1.25451$ & $2.74247$ & $2.09561$\\
      \hline
      \multirow{2}{*}{C} & \multirow{2}{*}{$0.26$} & \multirow{2}{*}{$1.4$} & \multirow{2}{*}{$\dfr{5}{3}$} & \multirow{2}{*}{ 1 }  &  \multirow{2}{*}{100}  & $1$ & $1$ &  $1$ & $0.6$ \\
        \cline{7-10}
                 &                        &                        &                               &                       &                      & $1.01595$ & $-0.93415$ & $0.24142$ & $0.15528$\\
      \hline
    \end{tabular}
  }
    \caption{Initial data of the three double shock Riemann problem.}
  \label{tab:acc-2-shock}
\end{table}
For all three cases the approximations by the structure-preserving scheme and by the vanishing viscosity scheme are compared.

\vspace{2mm}
\paragraph{Test A}

\begin{figure}[!htb]
\centering
\includegraphics[width=.45\linewidth]{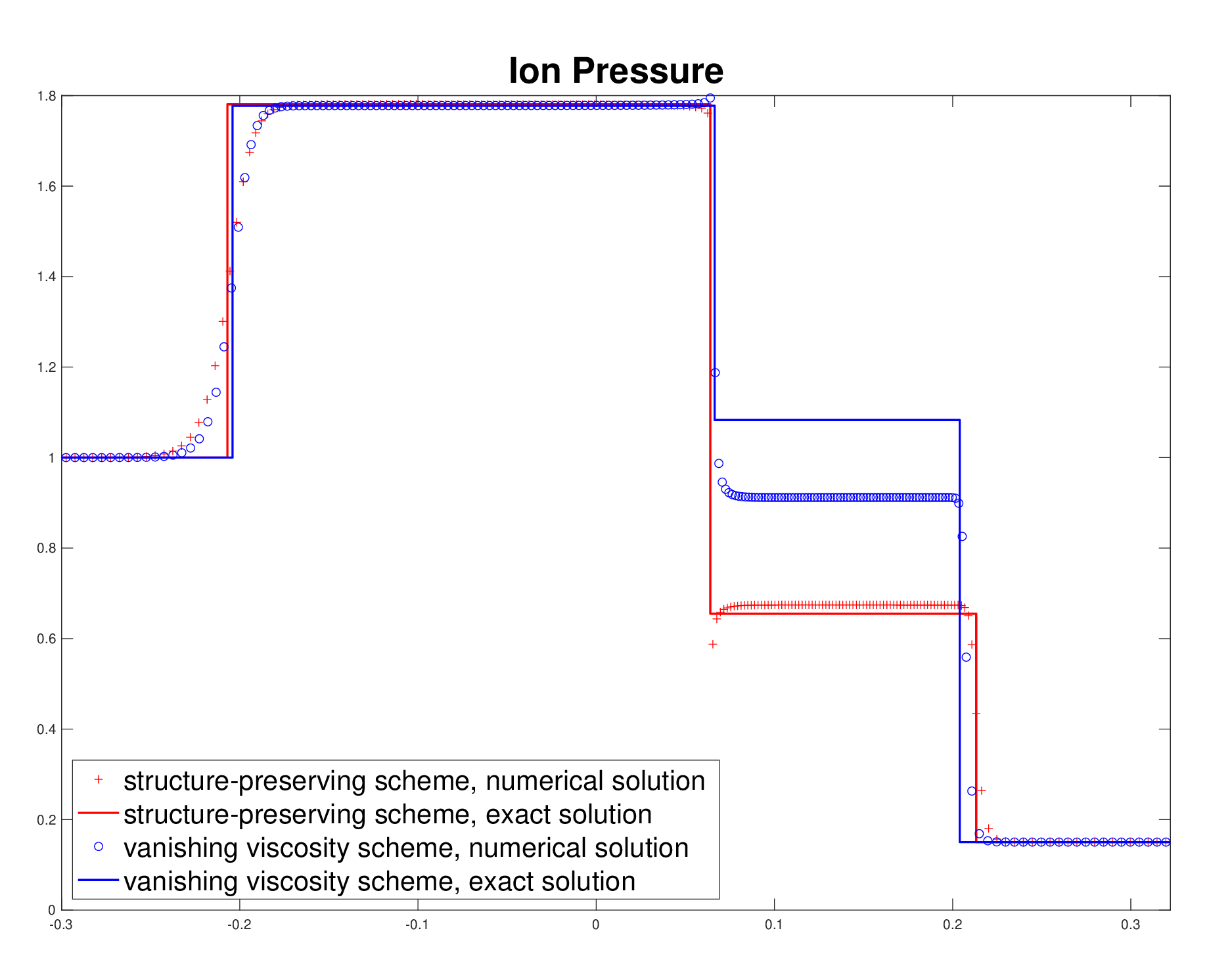}
\includegraphics[width=.45\linewidth]{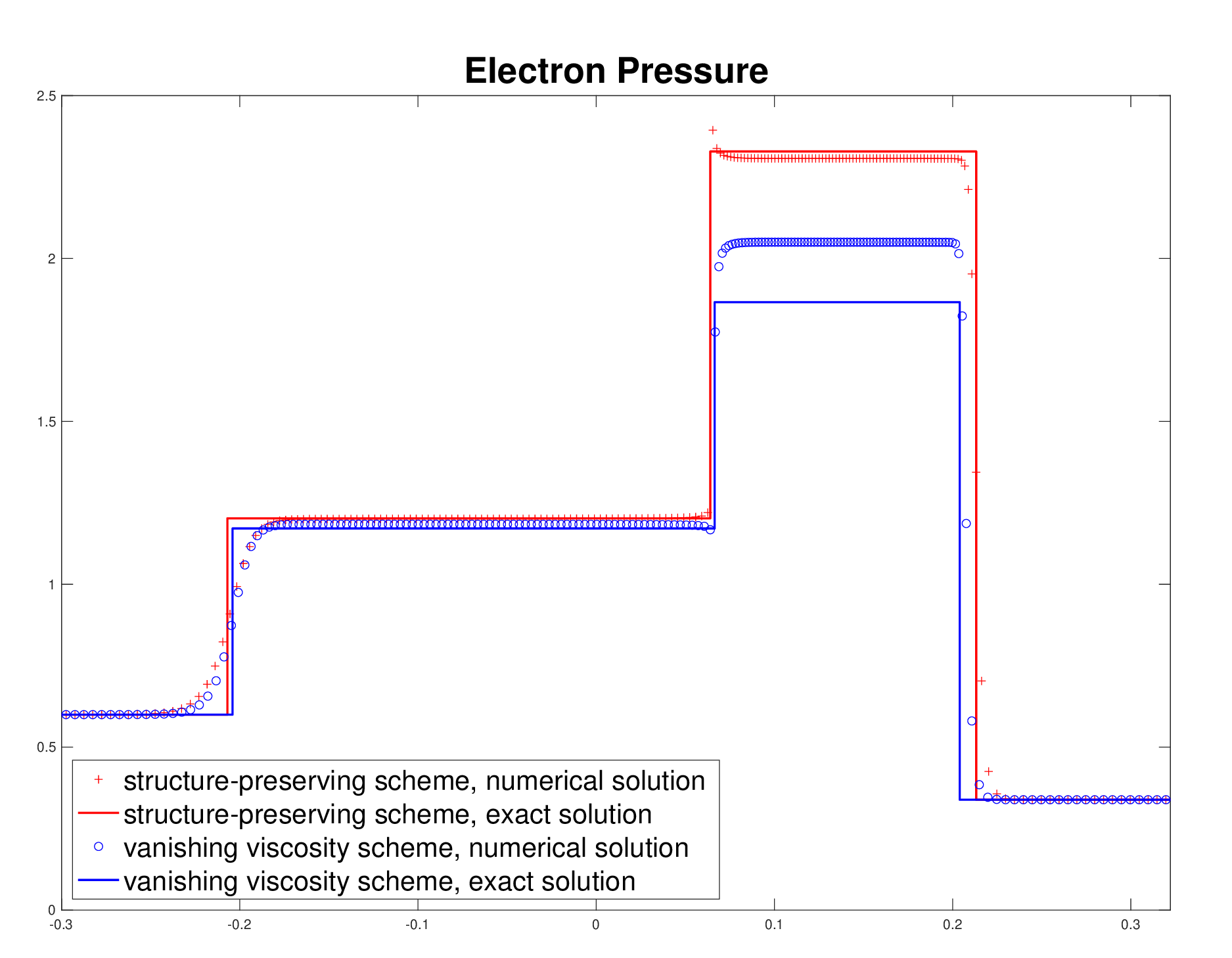}
\caption[small]{Test A. The ion pressure (left) and the electron pressure (right).}
\label{fig:A_p}
\end{figure}
For this Riemann problem involving shocks,  the two numerical schemes produce different solutions. The computed shock states reflect the Hugoniot relation associated with each scheme, as shown in Figure~\ref{fig:A_p}. Since the viscosity coefficient of ions is much larger than that of electrons, the vanishing viscosity scheme assumes a larger post-shock ion pressure.

\vspace{2mm}
\paragraph{Test B}

\begin{figure}[!htb]
\centering
\includegraphics[width=.45\linewidth]{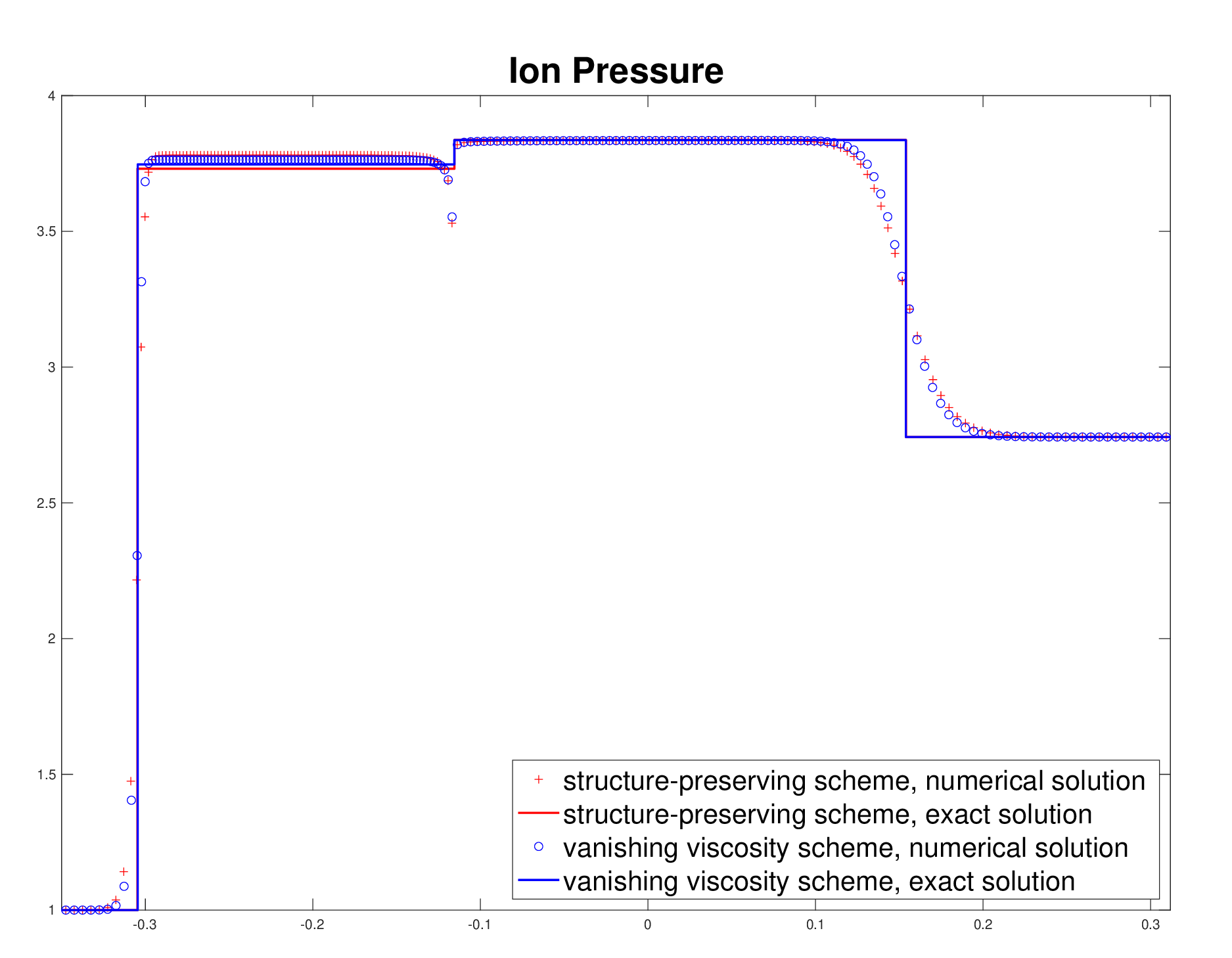}
\includegraphics[width=.45\linewidth]{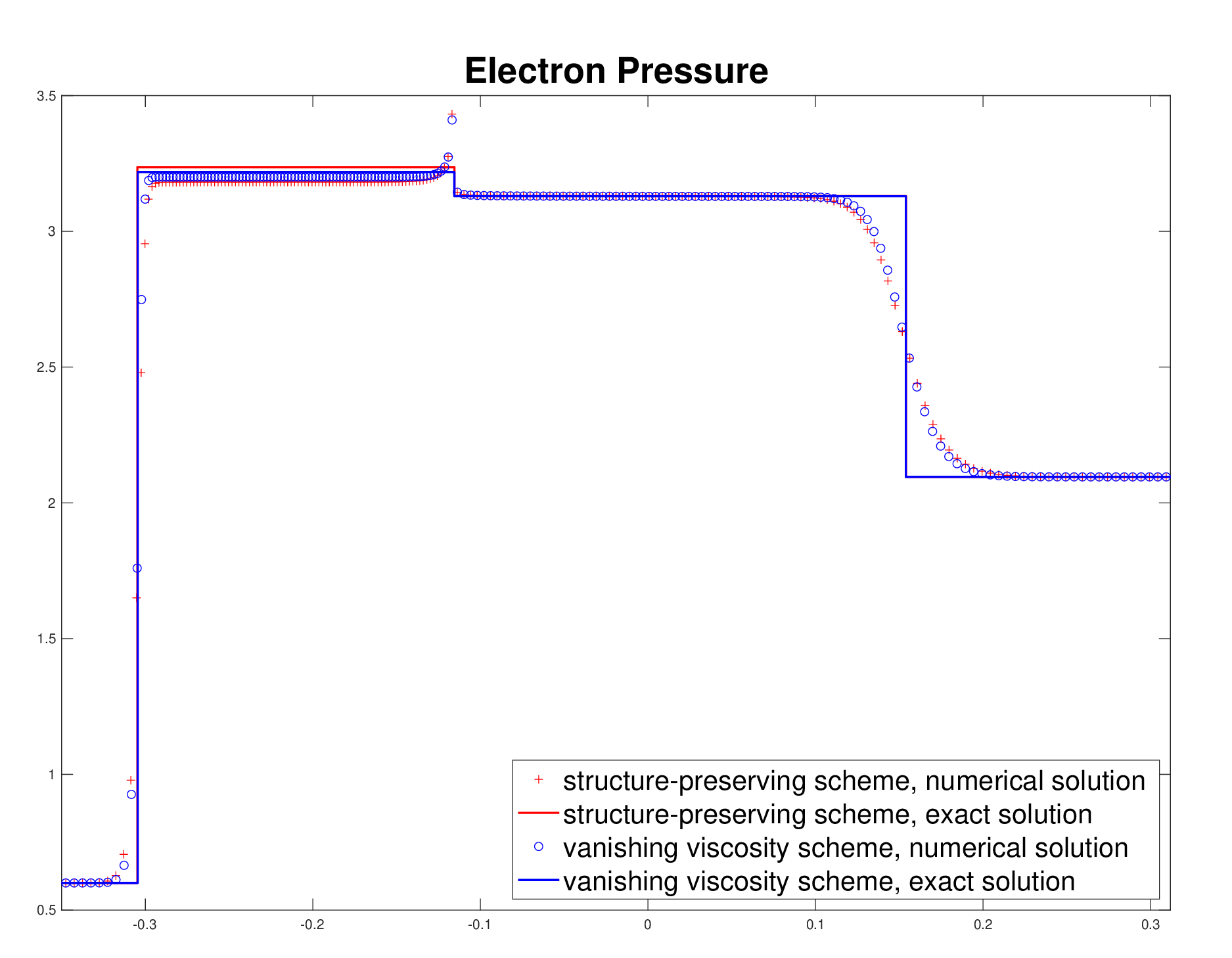}
\caption[small]{Test B. The ion pressure (left) and the electron pressure (right).}
\label{fig:B_p}
\end{figure}

Numerical solutions are made comparisons with corresponding exact solutions in Figure~\ref{fig:B_p}. Although the two numerical solutions look close to each other, they are not identical. This is due to the nonlinear dependence of the vanishing viscosity solution on the viscosity coefficients.

\vspace{2mm}
\paragraph{Test C}

\begin{figure}[!htb]
\centering
\includegraphics[width=.45\linewidth]{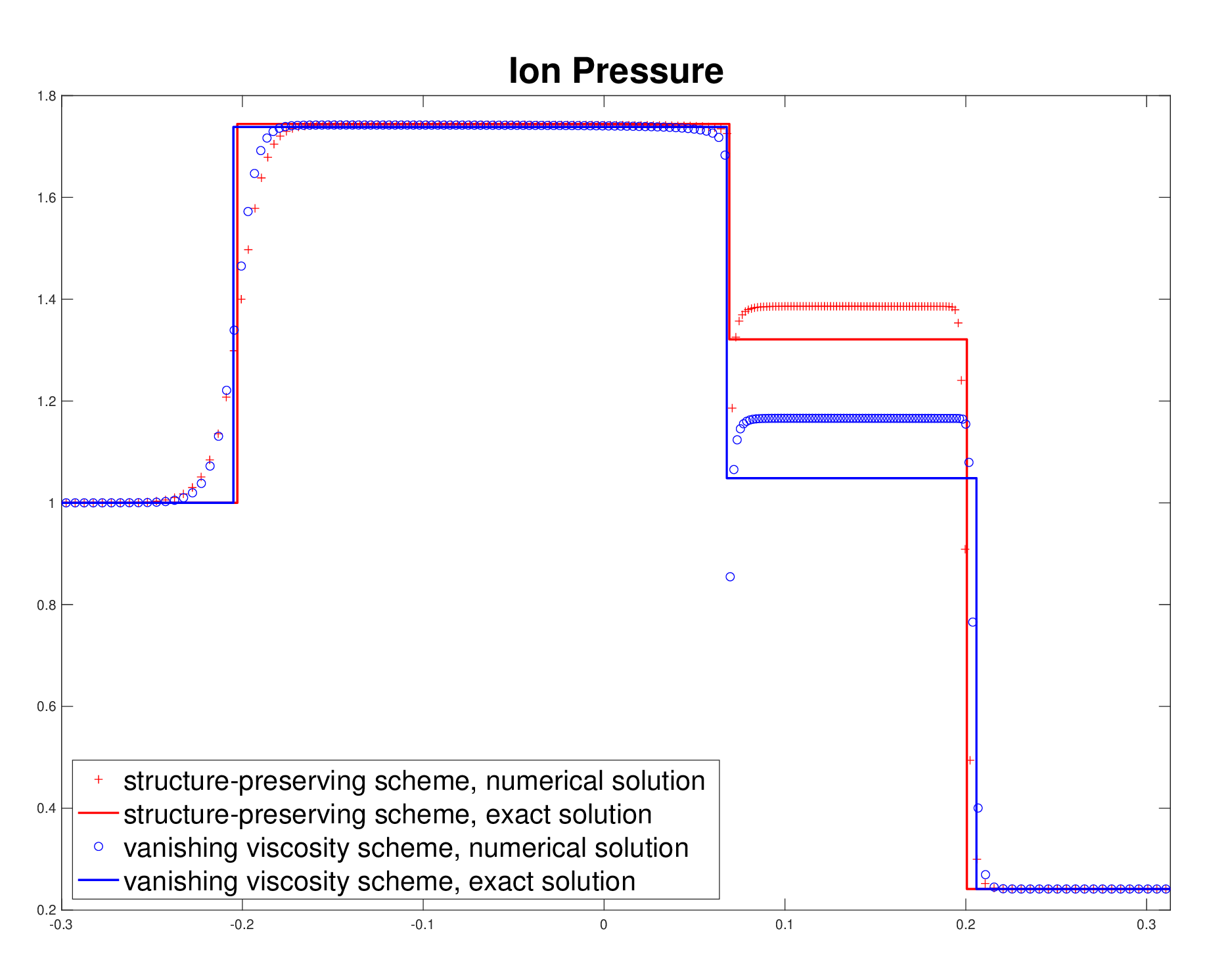}
\includegraphics[width=.45\linewidth]{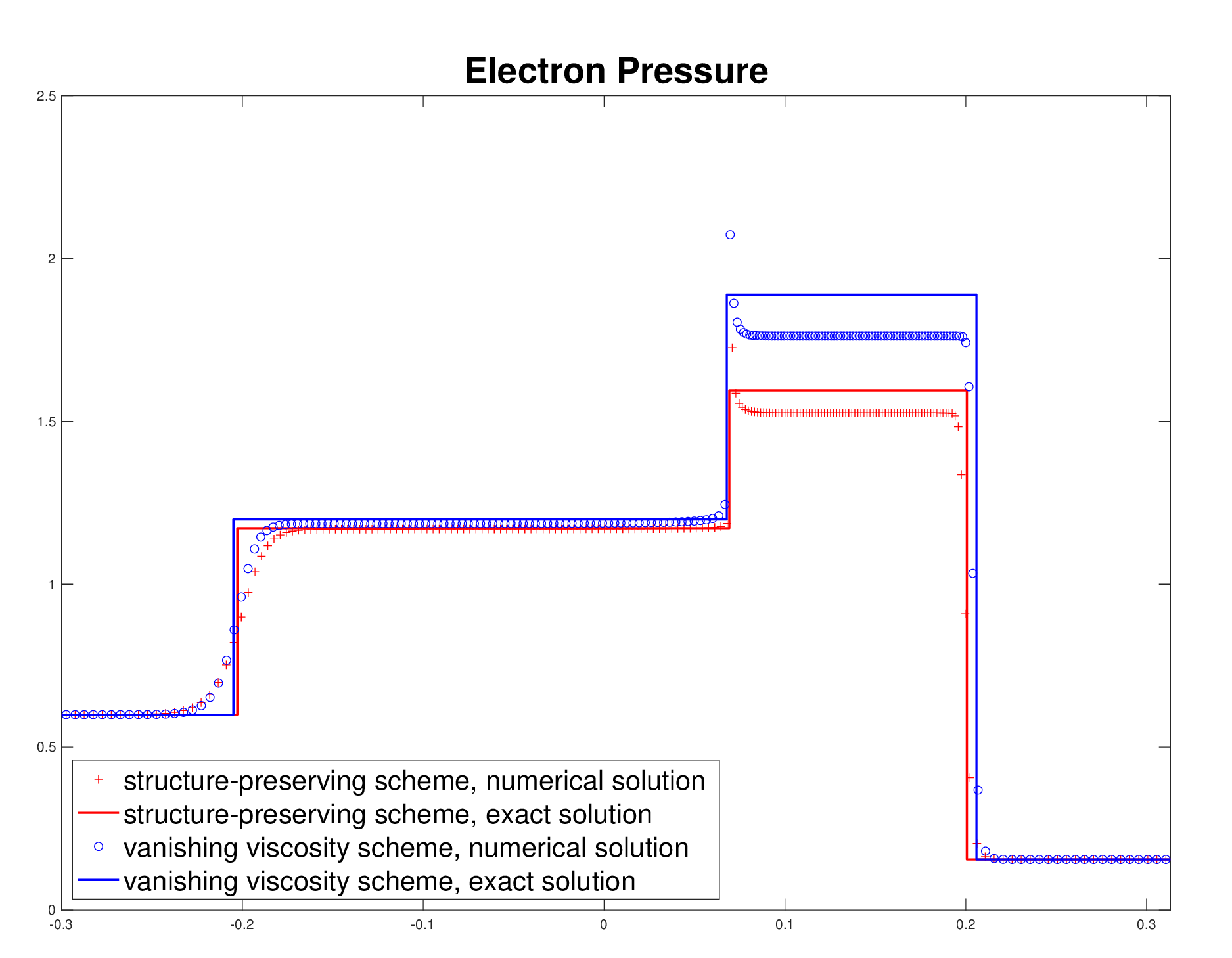}
\caption[small]{Test C. The ion pressure (left) and the electron pressure (right).}
\label{fig:C_p}
\end{figure}
In this case, the vanishing viscosity scheme results in a larger electron pressure since the electron viscosity coefficient is 100 times of the ion viscosity coefficient. This is clearly displayed in Figure ~\ref{fig:C_p}.

\vspace{2mm}
\subsubsection{Two-temperature double rarefaction wave Riemann problem}
The last case involves two rarefaction waves. The initial data is
\begin{equation}
(\rho, u, p_i, p_e)=\left\{
\bga{ll}
(1, \ -2, \ \frac 13, \ 0.2), & x < 0,\\[2mm]
(1, \  2, \ \frac 13, \ 0.2), & x > 0.
\eda
\right.
\end{equation}
The thermodynamical parameters are $\gm_i=1.4$, $\gm_e=\frac{5}{3}$, $\mu_i=1$, $\mu_e=100$. The final time is $0.1$. Both structure-preserving scheme and vanishing viscosity scheme are performed. 

\begin{figure}[!htb]
\centering
\includegraphics[width=.495\linewidth]{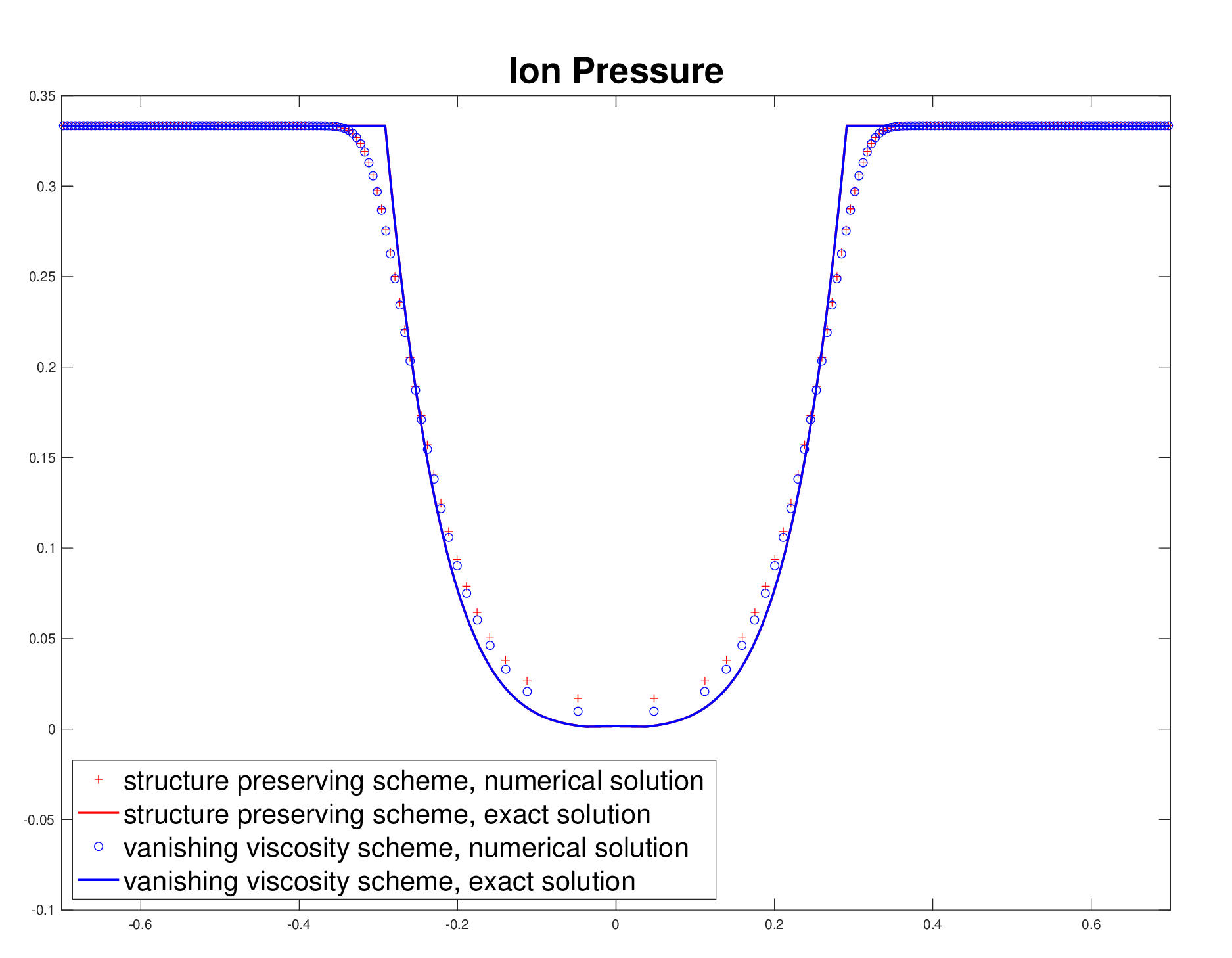}
\includegraphics[width=.495\linewidth]{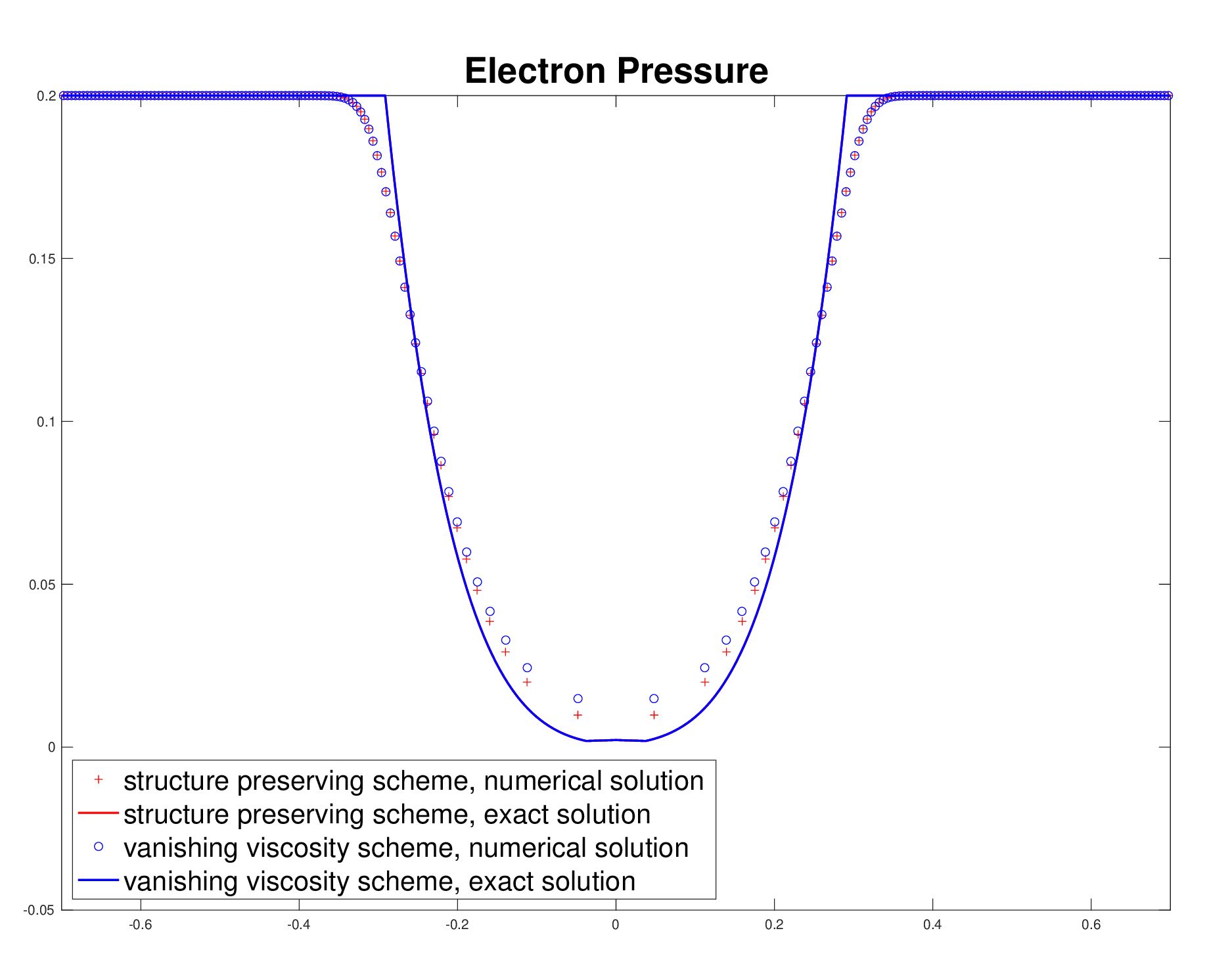}
\caption[small]{Double rarefaction. The ion pressure (left) and the electron pressure (right).}
\label{fig:D_p}
\end{figure}
Numerical solutions are displayed in Figure~\ref{fig:D_p}, in comparison with exact solutions.
As expected, for the smooth region in the rarefaction waves, both numerical schemes give the same result. Thus, the above results confirm that the numerical schemes correctly capture the structure of the Riemann problem solution for both the segment path and the vanishing viscosity Hugoniot.

\vspace{2mm}
\section{Discussion {and Outlook}}\label{sec:discussion}
The analysis in this paper shows that the ambiguity of shock solutions for multi-temperature Euler equations is not a numerical artifact, but a structural feature of the non-conservative model. The Rankine-Hugoniot relations for mass, momentum, and total energy determine only the Hugoniot submanifold of energetically admissible post-shock states. In contrast, a physically meaningful shock solution requires an additional selection of a curve on this hypersurface. In the language of the DLM theory~\cite{dal1995}, this selection is equivalent to prescribing a path for the non-conservative product; in the thermodynamic language used here, it is the choice of a Hugoniot relation that closes the shock problem.

The two Hugoniot relations studied in this paper make this point explicit. The segment-path relation in Section~\ref{sec:segment-path-hugoniot} splits the total Hugoniot relation into phasic contributions and gives a natural extension of the classical Courant--Friedrichs picture~\cite{Courant-Friedrichs} to the multi-temperature setting. The vanishing viscosity relation in Section~\ref{sec:vanishing-visc-hugoniot}, on the other hand, selects the shock states through traveling wave profiles and distributes the shock heating according to the viscosity coefficients. Both constructions conserve total energy, contact the isentropic curve with the expected order, and produce entropy along admissible shock branches, yet they generally lead to different ion and electron post-shock thermodynamic states. Thus, admissibility alone does not remove the physical non-uniqueness.

This distinction also clarifies the role of numerical methods. The structure-preserving scheme of~\cite{strt-prsv} formally encodes the segment-path Hugoniot relation, while the vanishing viscosity scheme of~\cite{scheme-e-vv} encodes the viscosity-weighted relation. The numerical experiments in Section~\ref{sec:numer} are consistent with this interpretation. In the double-shock tests, the two schemes produce different shock states precisely where the missing path information is active; in the double-rarefaction test, where the solution is smooth away from the fan structure and no shock path must be selected, the two schemes agree. A comparison between numerical approximations is therefore meaningful only after the intended Hugoniot relation has been identified.

Taken together, these observations suggest that the Hugoniot relation should be regarded as part of the physical closure of a multi-temperature flow model, rather than as a consequence of the macroscopic PDEs alone. For plasma flows, the appropriate closure depends on microscopic shock-layer physics and may have to be inferred from experiments~\cite{yarger-1955, mcqueen-1957, mcqueen-1960, isbell-1965, mcqueen-optical, royal-society-2013, caep-ifp-2013}, kinetic descriptions, or first-principles simulations~\cite{mattson2010first,swift2001first}. The present results provide reference Hugoniot relations and Riemann solvers against which such closures and their numerical discretizations can be tested. This viewpoint is fully consistent with the DLM framework, which provides a rigorous way to define non-conservative products once an admissible family of paths has been prescribed. At the same time, it emphasizes that the path should not be viewed merely as an auxiliary mathematical device: in physically derived models, it carries closure information that was lost when the microscopic description was reduced to a macroscopic non-conservative system.

Although the calculations were carried out for multi-temperature plasma flows, the mechanism revealed here is characteristic of general non-conservative hyperbolic PDEs. For a conservative system, the jump conditions across a discontinuity are fixed by the fluxes and the equation of state. For a non-conservative system, however, the product involving discontinuous quantities is not defined by the macroscopic equations alone. A weak formulation therefore requires additional path information, kinetic relations, viscosity profiles, relaxation mechanisms, or other microscopic closures. The present model gives a concrete example in which the missing information can be seen geometrically: the macroscopic jump conditions determine a surface of admissible states, while the unresolved physics selects a curve on that surface.

{Two directions are natural for future investigation. First, experimental or first-principles shock data can be used to construct data-informed Hugoniot closures and then to test whether the resulting macroscopic solutions reproduce the measured states. Second, the selected closures can be incorporated into high-order numerical schemes such as discontinuous Galerkin methods with positivity preservation and analyzed in the sharp-interface and vanishing-regularization limits.}

\bibliographystyle{siamplain}
\bibliography{ref}

\appendix
\vspace{2mm}

\section{Eigenvector structure}\label{app:eigen}

The left eigenvectors of the projected system in primitive variables~\eqref{eq:multi-temp-nD-primitive}, corresponding to the right eigenvectors~\eqref{eq:multi-temp-nD-r-eigenvec} are
\begin{align}\label{eq:multi-temp-nD-l-eigenvec}
\begin{split}
    \bl_{\pm} \coloneq \begin{bmatrix}
        0 \\
        \pm\frac{\rho}{2c}\bn\\
        \frac{1}{2c^2} \mathbf{1}_K
    \end{bmatrix}^T,\quad
        \bl_{0} \coloneq \begin{bmatrix}
        1 \\
        \mathbf{0}_d\\
       -\frac{1}{c^2}  \mathbf{1}_K 
    \end{bmatrix}^T, \quad
    \bl_{i} \coloneq \begin{bmatrix}
        0 \\
        \bt_i\\
        \mathbf{0}_K
    \end{bmatrix}^T, \\
    \bl_k \coloneq
    \frac{1}{k(k+1)}\begin{bmatrix}
       -(\bc^2)^T \cdot \sum_{j=1}^k \be_{K, j} - k \be_{K, k+1}\\
       \mathbf{0}_d \\
       \sum_{j=1}^k \be_{K, j} - k \be_{K, k+1}
    \end{bmatrix}^T,
    \end{split}
\end{align}
for $i = {1, \ldots, d-1}$ and $k={1, \ldots, K-1}$. Denoting the matrix of all zeros by $\mathbf{0}_{m\times n}  \in \mathbb{R}^{m \times n}$, the matrices $\bJ$ and $\bJ^{-1}$ are computed as
\begin{equation*}
    \bJ(\bU) = \begin{bmatrix}
        1 & \mathbf{0}_d^T & \mathbf{0}_K^T\\
        -\frac{1}{U_1^2}\begin{bmatrix}
                            U_2 \\ \vdots \\ U_{d+1}
                        \end{bmatrix} 
                        & \frac{1}{U_1} \mathbf{I}_d & \mathbf{0}_{d\times K}\\
        \mathbf{0}_K & \mathbf{0}_{K\times d} & \begin{bmatrix}
                                         \gamma_1 - 1 & &\\
                                        &\ddots &\\
                                        & &\gamma_K-1 
                                        \end{bmatrix}
    \end{bmatrix}
    = \begin{bmatrix}
        1 & \\
        -u_1 / \rho & 1 / \rho\\
        \vdots & & \ddots \\
        -u_d / \rho & & & 1 / \rho \\
        0  & & & & \gamma_1 - 1\\
        \vdots & & & & & \ddots\\
        0 & & & & & & \gamma_K-1
    \end{bmatrix}
\end{equation*}

and 
\begin{equation*}
    \bJ^{-1}(\bU) 
    = \begin{bmatrix}
        1 & \\
        u_1 &  \rho\\
        \vdots & & \ddots \\
        u_d  & & &  \rho \\
        0  & & & & 1/ (\gamma_1 - 1)\\
        \vdots & & & & & \ddots\\
        0 & & & & & & 1 / (\gamma_K-1)
    \end{bmatrix}.
\end{equation*}

\section{Regularity estimates for the vanishing viscosity terms}\label{app:regularity}
All differentiations in this subsection are taken along Hugoniot curves.

To perform the limiting procedure, it is important to prove that integrals in \eqref{eq:ds-1-vv} and \eqref{eq:ds-2-vv} are well-defined so they converge to zero as $\tau\rightarrow\tau_-$.
Since the integral curves of ODEs \eqref{eq:ode-ion-vv} is continuous at the initial point and we are concern with the $\tau$ close enough to $\tau_-$ here, we have $\pi_k^\text{vis}-p_{k,-}=\mathcal{O}(|\tau_--\tau|)$. Since $\tau$ and $m^2$ are two independent variables of the ODE integral curves,
\beqs
\dfr{\pt^\nu\pi_k^\text{vis}}{[\pt(m^2)]^\nu}=\mathcal{O}(|\tau_--\tau|),
\ \ \text{for }\forall \nu\in\mathbb{N}_+ \text{ and } k=i,e.
\eeqs
By \eqref{eq:ds-1-vv}, $\frac{ds_k}{d\tau}=\mathcal{O}(|\tau_--\tau|)$.

By the L'Hospital rule,
\beqs
m^2=\dfr{p-p_-}{\tau_--\tau}\rightarrow-\dfr{dp}{d\tau}(\tau_-),
 \ \ \text{as }\tau\rightarrow\tau_-.
\eeqs
So
\beqs
\dfr{dp}{d\tau}+m^2=\mathcal{O}(|\tau_--\tau|),
\eeqs
leading to
\beqs
\dfr{dm^2}{d\tau}=\dfr{1}{\tau_--\tau}\left(\dfr{dp}{d\tau}+m^2\right)=\mathcal{O}(1),
\eeqs
and
\beqs
\dfr{d^2m^2}{d\tau^2}=\dfr{2}{(\tau_--\tau)^2}\left(\dfr{dp}{d\tau}+m^2\right)+\dfr{1}{\tau_--\tau}\dfr{dp}{d\tau}=\mathcal{O}(|\tau_--\tau|^{-1}).
\eeqs
So
\beqs
\dfr{\pt\pi_k^\text{vis}}{\pt(m^2)}\dfr{dm^2}{d\tau}=\mathcal{O}(|\tau_--\tau|), \ \ \ 
\dfr{\pt^2\pi_k^\text{vis}}{[\pt(m^2)]^2}\left(\dfr{dm^2}{d\tau}\right)^2=
\dfr{\pt\pi_k^\text{vis}}{\pt(m^2)}\dfr{d^2m^2}{d\tau^2}=\mathcal{O}(1).
\eeqs
So, integrals in \eqref{eq:ds-1-vv} and \eqref{eq:ds-2-vv} goes to zero as $\tau_-\rightarrow\tau$.

\vspace{2mm}
{
\section{Formal Reduction of the Mixture EOS in the Thermal Equilibrium Limit}\label{app:thermal-eq-eos}
To demonstrate the fallback of the multi-temperature Euler equations in the thermal equilibrium limit, we first show that the mixture reduces to a classical single-fluid thermodynamic system for a general equation of state (EOS), before specifying the polytropic gas closure.

\subsection{General Equation of State}
For a general multi-temperature plasma, the thermodynamic state of each species $k$ can be parameterized by the mixture density $\rho$ and the phasic temperature $T_k$. Thus, the phasic internal energy and pressure are given by general constitutive relations $e_k = e_k(\rho, T_k)$ and $p_k = p_k(\rho, T_k)$. The phasic physical entropy $s_k$ satisfies the Gibbs relation:
\begin{equation}
    T_k ds_k = de_k + p_k d\tau,
\end{equation}
where $\tau = 1/\rho$ is the specific volume.

Under the assumption of thermal equilibrium, all species share a unified temperature, $T_1 = \dots = T_K = T$. Consequently, the total mixture internal energy and total mixture pressure become functions solely of $\rho$ and $T$:
\begin{align}
    e(\rho, T) &= \sum_{k=1}^K e_k(\rho, T), \\
    p(\rho, T) &= \sum_{k=1}^K p_k(\rho, T).
\end{align}
Provided that the mixture heat capacity at constant volume is strictly positive, $\left(\frac{\partial e}{\partial T}\right)_\rho > 0$, the internal energy is a monotonically increasing function of temperature. By the Inverse Function Theorem, this allows temperature to be uniquely determined by density and total internal energy, $T = T(\rho, e)$. Substituting this back into the total pressure yields a well-defined mixture equation of state completely determined by two independent thermal variables:
\begin{equation}
    p = p(\rho, T(\rho, e)) = p(\rho, e).
\end{equation}
Furthermore, substituting the equilibrium condition $T_k = T$ into the sum of the phasic Gibbs relations yields:
\begin{equation}
    T \sum_{k=1}^K ds_k = \sum_{k=1}^K de_k + \left(\sum_{k=1}^K p_k\right) d\tau.
\end{equation}
Defining the total mixture entropy as $s_\mathrm{phys} = \sum_{k=1}^K s_k$, this algebraically recovers the classical single-fluid Gibbs relation:
\begin{equation}
    T ds_\mathrm{phys} = de + p d\tau.
\end{equation}
This confirms that the mixture structurally behaves as a classical single-phase fluid governed by any two independent thermodynamic variables.

\subsection{Polytropic Gas Mixture}
For the specific multi-temperature polytropic gas model considered in this work, the phasic EOS is given by $p_k = (\gamma_k - 1)\rho e_k$. As shown in Section~\ref{sec:govern-pde}, the phasic temperature is identified as $T_k = p_k \tau = (\gamma_k - 1) e_k$.

Applying the formal thermal equilibrium constraint $T_1 = \dots = T_K = T$, the total mixture pressure and internal energy take the explicit forms:
\begin{align}
    p &= \sum_{k=1}^K p_k = K\rho T, \\
    e &= \sum_{k=1}^K e_k = T \sum_{k=1}^K \frac{1}{\gamma_k - 1}.
\end{align}
Isolating $T$ from the internal energy equation, $T = e \left( \sum_{k=1}^K \frac{1}{\gamma_k - 1} \right)^{-1}$, and substituting it into the pressure equation recovers a classical single-fluid perfect gas EOS:
\begin{equation}
    p = (\gamma_{\text{mix}} - 1)\rho e,
\end{equation}
where the effective adiabatic index of the mixture is uniquely defined as
\begin{equation}
    \gamma_{\text{mix}} = 1 + \frac{K}{\sum_{k=1}^K \frac{1}{\gamma_k - 1}}.
\end{equation}
This formal algebraic reduction provides structural justification for the geometric collapse of the Hugoniot hypersurface to a classical 1D Hugoniot curve, serving as a critical mathematical sanity check for the non-equilibrium framework.
}

\vspace{2mm}
{
\section{Proofs of Lemma \ref{lem:trajectory} and Theorem \ref{thm:entropy-vv}}
\vspace{2mm}
\subsection{Proof of Lemma \ref{lem:trajectory}}\label{app:lemma}
\begin{proof}
Consider the hypersurface
\beql{eq:critical-surface}
\mathcal{S}:=
\d\bigcup_{\tau<\tau_-}
\left\{(p_1,\dots,p_K,\tau) \ : 
\displaystyle\sum_{i=1}^K\frac{\gm_ip_i}{\tau}=m^2
\right\},
\eeq
What we need to prove is that for $(p_1,\dots,p_K)\in\mathcal{A}_1(\tau)\cap\mathcal{A}_2(\tau)$, the integral curve of \eqref{eq:ode-hugoniot-vv-pres} will (i) never go across the surface $\mathcal{S}$, and (ii) $\frac{dp_i}{d\tau}\leq\frac{d\mathfrak{q}_i}{d\tau}$ holds for $i=1,\dots,K$.

The normal vector at any point $(p_1^\mathcal{S},\dots,p_K^\mathcal{S},\tau)\in\mathcal{S}$  is
\beqs
\bn \coloneq
\left[
\bga{c}
\gm_1\tau_--(\gm_1+1)\tau\\[1.25mm]
\vdots\\[1.25mm]
\gm_K\tau_--(\gm_K+1)\tau\\[1.25mm]
-\sum_{k=1}^K(\gm_k+1)p_k^\mathcal{S}+p_-
\eda
\right],
\eeqs
which points to the direction where $\sum_{k=1}^K\frac{\gm_kp_k}{\tau}>m^2$. The tangential direction of the integral curve at $(p_1^\mathcal{S},\dots,p_K^\mathcal{S},\tau)$ is
\beql{eq:tangent-ode-vv}
\left.\bv\right|_{\mathcal{S}}=
\left[
\dfr{\gm_1p_1^\mathcal{S}}{\tau},
\dots,
\dfr{\gm_Kp_K^\mathcal{S}}{\tau},
-1
\right]^\top.
\eeq
Then we compute
\beql{eq:inner-product-oed-vv}
\bn\cdot\left.\bv\right|_{\mathcal{S}}=\dfr{\tau_--\tau}{\tau}\displaystyle\sum_{k=1}^K{\gm_k}^2p_k^\mathcal{S}
+\sum_{k=1}^K(p_k^\mathcal{S}-p_{k,-}),
\eeq
i.e.~$\bn\cdot\left.\bv\right|_{\mathcal{S}}>0$.
For any $(p_1,\dots,p_K,\tau)\in\mathcal{A}$, denote the tangential vector of the integral curve of \eqref{eq:ode-vv} by $\left.\bv\right|_\mathcal{A}$. If $\bn\cdot\left.\bv\right|_\mathcal{A}>0$, the local integral curve diverges from $\mathcal{S}$. Otherwise, there will be a point where $\bn\cdot\left.\bv\right|_\mathcal{A}=0$, where the integral curve will stop approaching $\mathcal{S}$.

\vspace{2mm}
Next, we show that for $\tau<\tau_-$,
\begin{equation}\label{eq:ineq-crit}
    \frac{dp_k}{d\tau}<\frac{d\mathfrak{q}_k}{d\tau},\,\,\, k=1,\dots,K.
\end{equation}
We first transform the integral definition of the vector $\widetilde{\mathbf{a}}=[\tilde{a}_{1},\,\dots,\,\tilde{a}_{K}]^\top$ into a system of linear first-order differential equations.
The vector $\widetilde{\mathbf{a}}(x)$ defined in \eqref{eq:def-tilde-a} is the unique solution to the initial value problem
\begin{equation*}
\frac{d\widetilde{\mathbf{a}}}{dx} = \mathcal{M} \widetilde{\mathbf{a}} + \boldsymbol{\kappa}, \quad \widetilde{\mathbf{a}}(0) = 0.
\end{equation*}
The constant vector $\boldsymbol{\kappa}=[\kappa_{1},\dots,\kappa_{K}]^\top$ has components $\kappa_{k} = \frac{\gamma_k}{\gamma_k-1}\frac{p_{k,-}}{\rho_-{c_-}^2}$, resembles the phasic bulk modulus.
Using the definition of the matrix $\mathcal{M}_{kj} = \gamma_k\delta_{kj} - \frac{\mu_k}{\mu}(\gamma_j-1)$, the $k$-th component of the differential equation is
\begin{equation}\label{eq:ode-tilde-a}
\frac{d\tilde{a}_k}{dx} - \gamma_k \tilde{a}_k = \kappa_{k} - \frac{\mu_k}{\mu} \Pi(x),
\end{equation}
where $\Pi(x) = \sum_{k=1}^K (\gamma_k - 1) \tilde{a}_k(x)$ is the thermodynamic coupling source.
We solve this linear first-order ODE by variation of constants. Multiply both sides by the integrating factor $e^{-\gamma_k x}$ to write is as
\begin{equation*}
\frac{d}{dx} \left( \tilde{a}_k(x) e^{-\gamma_k x} \right) = \left[ \kappa_{k} - \frac{\mu_k}{\mu} \Pi(x) \right] e^{-\gamma_k x},
\end{equation*}
which leads to
\begin{equation*}
\tilde{a}_k(x) e^{-\gamma_k x} 
= \frac{\kappa_{k}}{\gamma_k} (1 - e^{-\gamma_k x}) - \int_0^x \frac{\mu_k}{\mu} \Pi(y) e^{-\gamma_k y} \,dy.
\end{equation*}
Substitute $\kappa_{k} = \frac{\gamma_k}{\gamma_k-1}\frac{p_{k,-}}{\rho_-{c_-}^2}$ and $x=-\ln{\frac{\tau}{\tau_-}}$ into the above result to get
\begin{equation*}
\tilde{a}_k(x) = \frac{1}{\gamma_k - 1} \frac{p_{k,-}}{\rho_- c_-^2} \left[\left(\frac{\tau_-}{\tau}\right)^{\gamma_k} - 1\right] - \frac{\mu_k}{\mu} \int_0^x \Pi(y) e^{\gamma_k(x-y)} \,dy
\end{equation*}
Substitute this expression for $a_k(x)$ into the definition of $\mathfrak{q}_k= p_{k,-} + (\gamma_k - 1) \rho_- c_-^2 \tilde{a}_k$,
\begin{equation*}
\begin{aligned}
\mathfrak{q}_k &= (\gamma_k - 1) \rho_- c_-^2 \left[ \frac{1}{\gamma_k - 1} \frac{p_{k,-}}{\rho_- c_-^2} \left(\frac{\tau_-}{\tau}\right)^{\gamma_k} - \frac{\mu_k}{\mu} \int_0^x \Pi(y) e^{\gamma_k(x-y)} \,dy \right]\\
&=\mathfrak{p}_k -(\gamma_k - 1) \rho_- c_-^2 \frac{\mu_k}{\mu} \int_0^x \Pi(y) e^{\gamma_k(x-y)} \,dy.
\end{aligned}
\end{equation*}
So
\begin{equation*}
\begin{aligned}
\displaystyle\frac{d\mathfrak{q}_k}{d\tau} 
&=\displaystyle\frac{d\mathfrak{p}_k}{d\tau}
 +(\gamma_k - 1)\frac{\rho_- c_-^2}{\tau} \frac{\mu_k}{\mu}\left[
   \Pi(x)
 +\gamma_k \int_0^x \Pi(y) e^{\gamma_k(x-y)} \,dy
 \right].
\end{aligned}
\end{equation*}
By defining $\eta_k = \frac{p_{k,-}}{\sum_{j=1}^K \gamma_j p_{j,-}}$,
\begin{equation*}
\Pi(x) = \sum_{k=1}^K\eta_k\Pi_k(x).
\end{equation*}
The component $\Pi_k$ follows the Volterra integral equation
\begin{equation*}
\Pi_k(x) = (e^{\gamma_k x} - 1) - \int_0^x \Pi_k(y) \chi(x-y) \,dy,
\end{equation*}
where the kernel is $\chi(x) = \sum_{l=1}^K \frac{\mu_l}{\mu} (\gamma_l - 1) e^{\gamma_l x}$.
At $x=0$, $\Pi_k(0) = 0$ and the initial derivative is $\Pi^\prime_k(0) = \gamma_k$.\\

Note that $\Pi(x)$ depends on the pre-shock species pressure $p_{k,-}$ only through the coefficient $\eta_k$ and each compoenent $\Pi_k(x)$ only depends on the species adiabatic constant $\gamma_k$. So there exist a universal $\epsilon^*\in[0,1)$ independent of $(p_{1,-},p_{N,-},\tau_-)$ such that for $\tau\in(\epsilon^*\tau_-,\tau_-)$, $\Pi(x)>0$.

Finally, for $(p_1,\dots,p_K)\in\mathcal{A}_1(\tau)\cap\mathcal{A}_2(\tau)$ we recover $\frac{dp_k}{d\tau}<\frac{d\mathfrak{q}_k}{d\tau}$, for $k=1,\dots,K$.
The final result of this lemma is therefore proved.
\end{proof}

\vspace{2mm}
\subsection{Proof of Theorem \ref{thm:entropy-vv}}\label{app:thm-1}
\begin{proof}
Consider the simplex $\Xi(\tau)\in\mathbb{R}^{K}$ surrounded by hyperplanes
\begin{equation*}
\displaystyle\sum_{i=1}^K\frac{\gm_ip_i}{\tau}=\frac{p-p_-}{\tau_--\tau},
\end{equation*}
and
\begin{equation*}
p_k=p_{k,-},\,\,\, k=1\dots,K.
\end{equation*}
Its vertices are $\{\mathbf{v}_k:k=1,\dots,K\}$. If $j\neq k$, $\mathbf{v}_{k,j}=p_{j,-}$ and
\begin{equation*}
\mathbf{v}_{k,k}(\tau) = \frac{(p_{k,-} - \Sigma_{\neq k})\tau + \tau_- \Sigma_{\neq k}}{(1 + \gamma_k)\tau - \gamma_k \tau_-},
\end{equation*}
where
\begin{equation*}
\Sigma = \sum_{j=1}^K \gamma_j p_{j,-},\,\,\text{and}\,\,\Sigma_{\neq k} = \sum_{j \neq k} \gamma_j p_{j,-}.
\end{equation*}
One easily compute that
\begin{equation*}
\mathbf{v}^\prime_{k,k}(\tau_-) = -\frac{\Sigma}{\tau_-} = -\frac{1}{\tau_-} \sum_{j=1}^K \gamma_j p_{j,-}.
\end{equation*}
Since $\frac{dp_{k}}{d\tau}(\tau_-)=-\gamma_k\frac{p_{k,-}}{\tau_-}$, the trajectory of \eqref{eq:Hugoniot-vv} is tangent to the simplex at the initial state.
Next,
\begin{equation*}
\mathbf{v}^{\prime\prime}_{k,k}(\tau_-) = \frac{2(\gamma_k+1)}{\tau_-^2} \sum_{j=1}^K \gamma_j p_{j,-}.
\end{equation*}
Since $\frac{d^2p_{k}}{d\tau^2}(\tau_-)=\gamma_k(\gamma_k+1)\frac{p_{k,-}}{{\tau_-}^2}<\mathbf{v}^{\prime\prime}_{k,k}(\tau_-)$, $\exists\hat\tau<\tau$ such that $p_{k}(\tau)\in\mathcal{A}_1(\tau)$ for $\tau\in(\hat\tau,\tau_-)$.

By the proof of \eqref{eq:ineq-crit}, we know that $\frac{dp_k}{d\tau}\leq\frac{d\mathfrak{q}_k}{d\tau}$ for $\tau\leq\tau_-$ and the equality holds only at $\tau_-$. So $p_{k}(\tau)\in\mathcal{A}_2(\tau)$ for $\tau\in(\hat\tau,\tau_-)$.
\end{proof}
}
\end{document}

%% file: abbr-notation.tex
\def\Dt{\De t}

\def\ba{\mathbf{a}}
\def\bb{\mathbf{b}}
\def\bc{\mathbf{c}}

\def\be{\mathbf{e}}

\def\bl{\mathbf{l}}

\def\bn{\mathbf{n}}

\def\br{\mathbf{r}}

\def\bt{\mathbf{t}}
\def\bu{\mathbf{u}}
\def\bv{\mathbf{v}}

\def\bx{\mathbf{x}}

\def\bA{\mathbf{A}}
\def\bB{\mathbf{B}}
\def\bC{\mathbf{C}}

\def\bH{\mathbf{H}}
\def\bI{\mathbf{I}}
\def\bJ{\mathbf{J}}

\def\bL{\mathbf{L}}

\def\bR{\mathbf{R}}
\def\bS{\mathbf{S}}

\def\bU{\mathbf{U}}
\def\bV{\mathbf{V}}
\def\bW{\mathbf{W}}

\def\pt{\partial}

\def\f#1#2{\frac {#1}{#2}}
\def\f32{\frac 32}
\def\R{{\bf R}}

\def\d{\displaystyle}

\def\bga{\begin{array}}
\def\eda{\end{array}}

\def\De{\Delta}

\def\ev{\equiv}

\def\gm{\gamma}

\def\om{\omega}

\def\d{\displaystyle}
\def\dfr#1#2{\displaystyle{\frac{#1}{#2}}}

\DeclareMathOperator{\divop}{div}